\documentclass[conference,letterpaper]{IEEEtran}
\IEEEoverridecommandlockouts

\usepackage[utf8]{inputenc} 
\usepackage[T1]{fontenc}
\usepackage{ifthen}
\usepackage{cite}
\usepackage[cmex10]{amsmath} 

\usepackage{amsfonts}
\usepackage{amssymb}
\usepackage{amsthm}
\usepackage{enumitem}
\usepackage{xcolor}

\newtheorem{theorem}{Theorem}
\newtheorem{lemma}{Lemma}
\newtheorem{definition}{Definition}
\newtheorem{corollary}{Corollary}
\newtheorem{proposition}{Proposition}
\newtheorem{claim}{Claim}

\theoremstyle{definition}
\newtheorem{remark}{Remark}

\theoremstyle{definition} \newtheorem{example}{Example}

\def \extended {1}

\begin{document}

\title{Fractional Subadditivity of Submodular Functions: Equality Conditions and Their Applications\thanks{GK acknowledges support of ANRF, Govt. of India, under project  ANRF/ECRG/2024/005472/ENS. VP acknowledges support of DAE, Govt. of India, under project no. RTI4001.}}


\author{%
  \IEEEauthorblockN{Gunank Jakhar, Gowtham R. Kurri, Suryajith Chillara}
  \IEEEauthorblockA{International Institute of Information Technology, Hyderabad \\
                    Hyderabad, India.\\ 
                    Email: \{gunank.jakhar@research., gowtham.kurri@, suryajith.chillara@\}iiit.ac.in}
  \and
  \IEEEauthorblockN{Vinod M. Prabhakaran}
  \IEEEauthorblockA{Tata Institute of Fundamental Research\\
                    Mumbai, India.\\
                    Email: vinodmp@tifr.res.in}
}

\maketitle

\begin{abstract}
 Submodular functions are known to satisfy various forms of fractional subadditivity. This work investigates the conditions for equality to hold exactly or approximately in the fractional subadditivity of submodular functions. We establish that a small gap in the inequality implies that the function is close to being modular, and that the gap is zero if and only if the function is modular. We then present natural implications of these results for special cases of submodular functions, such as entropy, relative entropy, and matroid rank. As a consequence, we characterize the necessary and sufficient conditions for equality to hold in Shearer's lemma, recovering a result of Ellis~\emph{et~al.} (2016) as a special case. We leverage our results to propose a new multivariate mutual information, which generalizes Watanabe's total correlation (1960), Han's dual total correlation (1975), and Csisz\'ar and Narayan's shared information (2004), and analyze its properties. Among these properties, we extend Watanabe's characterization of total correlation as the maximum correlation over partitions to fractional partitions. When applied to matrix determinantal inequalities for positive definite matrices, our results recover the equality conditions of the classical determinantal inequalities of Hadamard, Sz\'asz, and Fischer as special cases.
\end{abstract}

\section{Introduction}
A submodular function is a set function that exhibits the property of \emph{diminishing returns}, i.e., the additional value gained by adding an element to a set decreases as the set grows larger~\cite{Fujishige05}. Some examples of submodular functions include entropy, matroid rank function, maximum cut in a graph. Submodular functions naturally arise in diverse areas such as machine learning~\cite{krause2007near}, combinatorial optimization~\cite{ConfortiC84}, algorithmic game theory~\cite{lehmann2001combinatorial}, social networks~\cite{kempe2003maximizing}, and statistical physics~\cite{VanEnterFS93}. 

It follows from the definition of submodular functions that they satisfy a property of subadditivity, provided that the function takes the value zero on the empty set. Specifically, the value of a submodular function on the ground set is less than or equal to the sum of the function values over any partition of the ground set. This property naturally generalizes to fractional partitions, where the ground set is covered by a family of overlapping subsets, with each element assigned to these subsets with some probability. This generalization, known as fractional subadditivity~\cite{ollagnier1982filtre,Feige06}, has many applications, including algorithmic game theory~\cite{Feige06,lehmann2001combinatorial} and statistical physics~\cite{VanEnterFS93}. Entropy, being an important example of a submodular function, satisfies fractional subadditivity~\cite{MadimanT10}, which, along with its special cases known as Shearer's lemma~\cite{ChungGFS86,Radhakrishnan03} and Han's inequality~\cite{Han78,Fujishige78}, has found applications in many areas such as graph theory, combinatorics, matrix inequalities, etc. \cite{Radhakrishnan03,MadimanT10,Sason22}.

In this work, we investigate the conditions for equality to hold exactly or approximately in the fractional subadditivity of submodular functions and explore their applications. Our main contributions are as follows:
\begin{itemize}
\item We show that a small inequality gap in the fractional subadditivity of submodular functions implies that the function is \emph{approximately} modular (Theorem~\ref{Stability_Submodular_Function_Partition}).
\item We establish that equality in the fractional subadditivity of a submodular function holds if and only if the function is modular (Theorem~\ref{Equality_Implies_Modularity_Partition}). This result enables the determination of whether a given submodular function $f$ is modular,  with only a minimal additional knowledge about $f$ (see Remark~\ref{remark:modularity}, and Examples~\ref{example1} and \ref{example2}).
\item We explore the implications of these results for specific submodular functions such as entropy, relative entropy, and matroid rank functions (Corollaries~\ref{Stability_Entropy_Partition}-\ref{Matroid_Rank_Equality}). Notably, we provide the necessary and sufficient conditions for equality to hold in Shearer's lemma~\cite[Lemma~1]{Radhakrishnan03} as a special case and this extends a result by Ellis \emph{et al.}~\cite[Lemma~9]{EllisKFY16}. To the best of our knowledge, these conditions, in this level of generality, do not appear explicitly in the literature. 
\item As an application, we propose a new multivariate mutual information measure, which generalizes total correlation~\cite{Watanabe60}, dual total correlation~\cite{Han75}, and shared information~\cite{CsiszarN04,narayan2016multiterminal} (Proposition~\ref{prop:MMI-specialcases}). The latter has operational significance as secret-key capacity for multiple terminals~\cite{CsiszarN04}. We analyze the properties of the proposed measure and extend Watanabe's characterization of total correlation as the maximum correlation over partitions~\cite{Watanabe60} to fractional partitions (Theorem~\ref{Multivariate_Mutual_Information}).
\item Additionally, we study applications of our results in matrix determinantal inequalities (Proposition~\ref{Determinantal_Equality}), recovering equality conditions for classical determinantal inequalities, including those of Hadamard, Sz\'asz, and Fischer~\cite{HornJ12,DemboCT91}.
\end{itemize}

\section{Preliminaries}
\textit{Notation}: We use $[i : i+k]$ to represent the set $\{i, i+1, \ldots, i+k-1, i+k\}$, where $i,k\in\mathbb{N}$. The power set of a set $A$ is denoted by $2^A$. $A^{\text{c}}$ denotes the complement of a set $A$. We use $\mathcal{F}$ to denote a family of subsets of $[1:n]$ allowing for repetitions, represented as $\{\!\!\{\cdot\}\!\!\}$.

\begin{definition}[Sub/Supermodular, and Modular Functions~\cite{Fujishige05}]
A set function $f: 2^{[1:n]} \rightarrow \mathbb{R}$ is called submodular if
\begin{align}
    f(S) + f(T) \geq f(S \cup T) + f(S \cap T),\ \forall S,T \subseteq [1:n].
\end{align}
A function $f: 2^{[1:n]} \rightarrow \mathbb{R}$ is called supermodular if $-f$ is submodular, and modular if it is both submodular and supermodular. Moreover, if 
$f(\phi) = 0,~f$ is modular if and only if $f(A)=\sum_{i\in A}f(\{i\})$, $\forall A\subseteq [1:n]$.
\end{definition}

Madiman and Tetali~\cite{MadimanT10} demonstrated that submodular functions satisfy a conditioning property and a chain rule under an appropriately defined notion of conditioning. Specifically, for $S,T\subseteq [1:n]$, the conditional version of submodular function is defined as $f(S|T) = f(S\cup T) -f(T)$. Let $S,T,U$ be disjoint subsets of $[1:n]$, then \cite[Lemma~IV]{MadimanT10} states that
\begin{align}
      f(S | T, U) &\leq f(S | T), \label{eqref: conditioning}\\
       f([1:n]) &= \sum_{i=1}^n f(\{i\} | [1:i-1]), \label{eqref: chain_rule}
\end{align}
where $f(S|T,U):=f(S|T\cup U)$. We now recall the following related concepts from fractional graph theory~\cite{ScheinermanU97}, \cite[Definition~II]{MadimanT10} as we would need them for our results. 

\begin{definition}[Fractional Covering, Packing, and Partition]
    \hspace{0.5em}
    \begin{enumerate}[leftmargin=*]
         \item Given a family $\mathcal{F}$ of subsets of $[1:n]$, a function $\alpha:~\mathcal{F}\rightarrow\mathbb{R}_{+}$ is called a fractional covering, if for all $i \in [1:n]$, $\sum\limits_{S \in \mathcal{F}: i \in S} \alpha(S) \geq 1$ .
        \item Given a family $\mathcal{F}$ of subsets of $[1:n]$, a function $\beta:~\mathcal{F}\rightarrow\mathbb{R}_{+}$ is called a fractional packing, if for all $i \in [1:n]$, $\sum\limits_{S \in \mathcal{F}: i \in S} \beta(S) \leq 1$.
        \item If $\gamma:\mathcal{F}\rightarrow\mathbb{R}_{+}$ is both a fractional covering and a fractional packing, then it is called a fractional partition.
    \end{enumerate}
\end{definition}
In this paper, we investigate the exact and approximate equality conditions of inequalities stated in the following theorem.
\begin{theorem}[\!\!{\cite{MadimanT10}}\footnote{\eqref{eqn:MTIneq} is the weak-fractional form of \cite[Theorem~I]{MadimanT10}, as presented in \cite[Section~VII]{MadimanT10}.}]\label{MT inequality}

    Let $f: 2^{[1:n]} \rightarrow \mathbb{R}$ be any submodular function with $f(\phi) = 0$. Let $\gamma:\mathcal{F}\rightarrow \mathbb{R}_+$ be any fractional partition with respect to a family $\mathcal{F}$ of subsets of $[1:n]$. Then
    \begin{align}\label{eqn:MTIneq}
        \sum\limits_{S \in \mathcal{F}} \gamma(S)f(S | S^\emph{c}) \leq f([1:n]) \leq \sum\limits_{S \in \mathcal{F}} \gamma(S)f(S).
    \end{align}
The fractional partition $\gamma$ in the lower and upper bounds can be replaced by fractional packing $\beta$ and fractional covering $\alpha$, respectively, if the submodular function $f$ is such that $f([1:j])$ is non-decreasing in $j$ for $j\in[1:n]$.
\end{theorem}

The upper bound in \eqref{eqn:MTIneq}, with fractional covering $\alpha$, has been referred to as fractional subadditivity of submodular functions in the literature~\cite{Feige06}. There is a duality between the upper and lower bounds in \eqref{eqn:MTIneq}, relating the gaps in the inequalities. For any family $\mathcal{F}$, its complementary family is defined as $\bar{\mathcal{F}} = \{\!\!\{ S^{\text{c}}: S \in \mathcal{F}\}\!\!\}$. For a fractional partition $\gamma$, its dual fractional partition $\Bar{\gamma}$ is defined by $\Bar{\gamma}(S^{\text{c}}) = \frac{\gamma(S)}{w(\gamma) - 1}, \quad \forall S \in \mathcal{F}$, where $w(\gamma)$ denotes the weight of $\gamma$, given by $w(\gamma) = \sum_{S \in \mathcal{F}} \gamma(S)$.

\begin{theorem}[\!\!{\cite[Theorem~IV]{MadimanT10}}]\label{Duality}
    Let $f: 2^{[1:n]} \rightarrow \mathbb{R}$ be any submodular function with $f(\phi) = 0$. Let $\gamma:\mathcal{F}\rightarrow \mathbb{R}_+$ be any fractional partition with respect to a family $\mathcal{F}$ of subsets of $[1:n]$. Then,
    \begin{align}\label{eq: dual}
        \frac{\emph{Gap}_\emph{U}(f, \mathcal{F}, \gamma)}{w(\gamma)} = \frac{\emph{Gap}_\emph{L}(f, \Bar{\mathcal{F}}, \Bar{\gamma})}{w(\Bar{\gamma})},
    \end{align}
    where
    \begin{align}
        \emph{Gap}_\emph{L}(f, \mathcal{F}, \gamma) &= f([1:n]) - \sum\limits_{S \in \mathcal{F}} \gamma(S)f(S | S^{\emph{c}}),\ \text{and} \label{eq: lgap}\\
        \emph{Gap}_\emph{U}(f, \mathcal{F}, \gamma) &= \sum\limits_{S \in \mathcal{F}} \gamma(S)f(S) - f([1:n]). \label{eq: ugap}
    \end{align}
\end{theorem}

\section{Equality Conditions in Fractional Subadditivity}\label{section:stab-eql}
We begin by outlining the assumptions on family $\mathcal{F}$ of subsets of $[1:n]$ and the fractional partition $\gamma$, covering $\alpha$, and packing $\beta$ with respect to $\mathcal{F}$, which hold throughout this paper:
\begin{enumerate}[leftmargin=*]
    \item No two indices $i,j\in[1:n]$ always appear together in the members of $\mathcal{F}$.
    \item The family $\mathcal{F}$ includes only proper subsets of $[1:n]$, i.e., $[1:n]\notin\mathcal{F}$.
    \item For all $S\in\mathcal{F}$, we have $\gamma(S),\alpha(S),\beta(S)>0$.
\end{enumerate}

It is shown in
\if \extended 1%
Appendix \ref{appendix:a}
\fi
\if \extended 0%
the extended version \cite[Appendix~A]{JakharKCP25}
\fi
that the above assumptions hold without loss of generality in the context of the fractional subadditivity of submodular functions.

The following theorem provides approximate equality conditions for fractional subadditivity, showing that for a submodular function $f$, if $\text{Gap}_\text{L}(f,\mathcal{F},\gamma)$ or $\text{Gap}_{\text{U}}(f,\mathcal{F},\gamma)$ is \emph{small}, then $f$ is \emph{approximately} modular. 

\begin{theorem}\label{Stability_Submodular_Function_Partition}
    Let $f: 2^{[1:n]} \rightarrow \mathbb{R}$ be any submodular function with $f(\phi) = 0$, and let $\gamma:\mathcal{F}\rightarrow\mathbb{R}_+$ be any fractional partition with respect to a family $\mathcal{F}$ of subsets of $[1:n]$. For $\emph{Gap}_{\emph{L}}(f, \mathcal{F}, \gamma)$ and $\emph{Gap}_{\emph{U}}(f, \mathcal{F}, \gamma)$ as defined in \eqref{eq: lgap} and \eqref{eq: ugap}, respectively, and any $\varepsilon \geq 0$, the following holds:
    
        If $\emph{Gap}_{\emph{L}}(f, \mathcal{F}, \gamma) \leq  \varepsilon$ or $\emph{Gap}_{\emph{U}}(f, \mathcal{F}, \gamma) \leq  \varepsilon$, then
        \begin{align}
        \hspace{-0.2pt} f(\{i\}) + f([1:n] \setminus{\{i\}})  - f([1:n]) \leq \frac{\varepsilon}{\sigma},\  \forall i \in [1:n],
        \end{align}
      where $\sigma = \min\limits_{\substack{i, j \in [1:n]: \\ i \neq j}} \sum\limits_{\substack{ S \in \mathcal{F}: \\ i \in S, j \notin S }} \gamma(S)>0$.
\end{theorem}
\begin{remark}\label{remark:alphabeta}
    Although Theorem~\ref{Stability_Submodular_Function_Partition} is stated for a fractional partition $\gamma$, if the submodular function $f$ is such that $f([1:j])$ is non-decreasing in $j$ for $j\in[1:n]$, the theorem's assertions can be generalized to fractional covering $\alpha$ and fractional packing $\beta$ for $\text{Gap}_{\text{L}}(f,\mathcal{F},\beta)$ and $\text{Gap}_{\text{U}}(f,\mathcal{F},\alpha)$, respectively, in a straightforward manner. These details are presented towards the end of
    \if \extended 1%
        Appendix \ref{appendix:b}.
    \fi
    \if \extended 0%
        \cite[Appendix~B]{JakharKCP25}.
    \fi
    We have chosen $\gamma$ to state Theorem~\ref{Stability_Submodular_Function_Partition} because the quantity $\sigma$, as defined in the theorem, may not always be strictly positive when using $\alpha$ or $\beta$ instead.
\end{remark}
The proof of Theorem~\ref{Stability_Submodular_Function_Partition} refines the approach in the proof of Shearer's lemma by Llewellyn and Radhakrishnan~\cite{Radhakrishnan03}, and incorporates insights from its stability version by Ellis~\emph{et~al.}~\cite{EllisKFY16}. A detailed proof is provided in
\if \extended 1%
Appendix~\ref{appendix:b}.
\fi
\if \extended 0%
\cite[Appendix~B]{JakharKCP25}.
\fi

The next theorem states that equality in the fractional subadditivity of submodular functions holds if and only if $f$ is modular. 

\begin{theorem} \label{Equality_Implies_Modularity_Partition}
    Let $f: 2^{[1:n]} \rightarrow \mathbb{R}$ be any submodular function with $f(\phi) = 0$. Let $\gamma: \mathcal{F} \rightarrow \mathbb{R}_{+}$, $\alpha: \mathcal{F} \rightarrow \mathbb{Q}_{+}$, and $\beta:\mathcal{F}\rightarrow \mathbb{Q}_+$ be any fractional partition, fractional covering, and fractional packing with respect to a family $\mathcal{F}$ of subsets of $[1:n]$, respectively. Then the following hold:
    \begin{enumerate}[leftmargin=*]
        \item $\emph{Gap}_{\emph{U}}(f, \mathcal{F}, \gamma) =  0$ (similarly, $\emph{Gap}_{\emph{L}}(f, \mathcal{F}, \gamma) =  0$) if and only if $f$ is modular. 
        \item If $f$ is non-decreasing, i.e., $f(S)\leq f(T)$ for $S\subseteq T$, then $\emph{Gap}_{\emph{U}}(f, \mathcal{F}, \alpha) =  0$ (resp. $\emph{Gap}_{\emph{L}}(f, \mathcal{F}, \beta) =  0$) if and only if $f$ is modular and $f(Z)=0$, for all $Z\subseteq\{i\in[1:n]: \sum_{S \in \mathcal{F}: i \in S} \alpha(S) > 1 \}$ (resp. for all $Z\subseteq\{i\in[1:n]: \sum_{S \in \mathcal{F}: i \in S} \beta(S) < 1 \}$).
    \end{enumerate}
    
\end{theorem}
A detailed proof of Theorem~\ref{Equality_Implies_Modularity_Partition} is given in
\if \extended 1%
Appendix~\ref{appendix:c}.
\fi
\if \extended 0%
\cite[Appendix~C]{JakharKCP25}.
\fi
Part~1) of Theorem~\ref{Equality_Implies_Modularity_Partition} can be viewed as a corollary of Theorem~\ref{Stability_Submodular_Function_Partition}, but part~2) requires that $f$ is a non-decreasing set function. This condition is stronger than the condition for the inequality $\text{Gap}_{\text{U}}(f, \mathcal{F}, \alpha)\geq 0$ to hold, where $f([1:j])$ is assumed to be non-decreasing in $j$, for $j\in[1:n]$ (see Theorem~\ref{MT inequality}). Interestingly, this stronger condition is necessary\footnote{\if \extended 1%
Appendix~\ref{appendix:d}
\fi
\if \extended 0%
\cite[Appendix~D]{JakharKCP25}
\fi
provides an example of a submodular function that is not modular, where $f([1:j])$ is non-decreasing in $j$, $f(S)>f(T)$ for some $S\subseteq T$, and $\text{Gap}_{\text{U}}(f, \mathcal{F}, \alpha)= 0$.} for the proof of part 2) to go through.
\begin{remark}
    While duality in \eqref{eq: dual} establishes a relationship between $\text{Gap}_\text{U}(f, \mathcal{F}, \gamma)$  and $\text{Gap}_\text{L}(f, \Bar{\mathcal{F}}, \Bar{\gamma})$, it does not give a direct connection between $\text{Gap}_\text{L}(f, \mathcal{F}, \gamma)$ and $\text{Gap}_\text{U}(f, \mathcal{F}, \gamma)$ as noted in \cite[Section~VII]{MadimanT10}. Nevertheless, Theorem~\ref{Equality_Implies_Modularity_Partition} implies that $\text{Gap}_\text{L}(f, \mathcal{F}, \gamma)=0$ if and only if $\text{Gap}_\text{U}(f, \mathcal{F}, \gamma)=0$.
\end{remark}
\begin{remark}\label{remark:modularity}
    Theorem~\ref{Equality_Implies_Modularity_Partition} can be used to determine whether a submodular function $f$ is in fact a modular function, with minimal additional knowledge about $f$. Let us suppose that the values of a submodular function $f$ are known for a family $\mathcal{F}$ of subsets of $[1:n]$, which admits a fractional partition, and for $[1:n]$. For instance, if the family $\mathcal{F}=\{\!\!\{S\subseteq [1:n]: |S|=k\}\!\!\}$, for some $k\in[1:n]$, then $\gamma(S)= 1/\binom{n-1}{k-1}$, $\forall S\in\mathcal{F}$ defines a fractional partition. Even without the knowledge of the values of $f$ on the remaining subsets of $[1:n]$, i.e., on $2^{[1:n]}\setminus (\mathcal{F} \cup [1:n])$, we can conclude that $f$ is modular if and only if $\sum_{S\in\mathcal{F}}\gamma(S)f(S)=f([1:n])$.
    This is illustrated in the examples that follow. 
\end{remark}
\begin{example}[Modularity with a symmetric\footnote{We call a family $\mathcal{F}$ symmetric if it remains invariant under every permutation on $[1:n]$.} $\mathcal{F}$]\label{example1}
    Suppose $f:2^{[1:4]}\rightarrow \mathbb{R}$ is a submodular function and the values of $f$ on the members of $\mathcal{F}=\{\!\!\{\{i\}:i\in[1:4]\}\!\!\}$ and for $[1:4]$ are given as $f(\{i\})=i\cdot 2^i$, for $i\in[1:4]$, and $f([1:4])=98$. Note that $\gamma:\mathcal{F}\rightarrow \mathbb{R}_+$, defined by $\gamma(S)=1$, $\forall S\in\mathcal{F}$, forms a fractional partition. Since $\sum_{i=1}^4f(\{i\})=\sum_{i=1}^4 i\cdot 2^i=98=f([1:4])$, Theorem~\ref{Equality_Implies_Modularity_Partition} implies that $f$ is a modular function. This conclusion holds irrespective of the values of $f$ on $2^{[1:4]}\setminus (\mathcal{F}\cup [1:4])$.
\end{example}
\begin{example}[Modularity with a non-symmetric $\mathcal{F}$]\label{example2}
   Let $f:2^{[1:4]}\rightarrow \mathbb{R}$ be a submodular function, and the values of $f$ on members of a family $\mathcal{F}$ and on $[1:4]$ are given as: $f(\{2\})=0.3$, $f(\{4\})=5$, $f(\{1,2\})=3$, $f(\{3,4\})=0.6, f(\{2,4\})=5.3, f(\{1,2,3\})=-1.4, f(\{1,3,4\})=3.3$ and $f(\{1,2,3,4\})=3.6$. Note that $\gamma:\mathcal{F}\rightarrow \mathbb{R}_+$, defined by $\gamma(\{2\})=\frac{1}{6}$, $\gamma(\{4\})=\frac{5}{12}$, $\gamma(\{1,2\})=\frac{1}{12}$, $\gamma(\{3,4\})=\frac{1}{12}$, $\gamma(\{2,4\})=\frac{1}{6}$, $\gamma(\{1,2,3\})=\frac{7}{12}$, $\gamma(\{1,3,4\})=\frac{1}{3}$, forms a fractional partition. Since $\sum_{S\in\mathcal{F}}\gamma(S)f(S)= 3.6 =f([1:4])$, Theorem~\ref{Equality_Implies_Modularity_Partition} implies that $f$ is a modular function. Note that this conclusion holds even though $f$ takes negative values for some subsets.  
\end{example}
Identifying the modularity of a submodular function can significantly improve the performance guarantees of optimization algorithms. For example, for certain combinatorial optimization problems, the greedy algorithm is guaranteed to find an exact optimal solution if the submodular function is, in fact, modular~\cite{ConfortiC84}. In contrast, for general submodular functions, it typically achieves only an approximately optimal solution \cite{FisherNW78,ConfortiC84}.

\section{Implications for Entropy, Relative Entropy, and Matroid Rank}
In this section, we present the implications of our results from Section~\ref{section:stab-eql} for specific submodular functions: entropy, relative entropy, and matroid rank function.  
\subsection{Implications for Entropy}
Consider jointly distributed random variables $X_1,X_2,\ldots,X_n$. It is well-known that the function $\mathsf{e}(F)=H(X_F)$, where $H(\cdot)$ denotes entropy and $F\subseteq [1:n]$, is submodular~\cite{Fujishige05}.

\begin{corollary}\label{Stability_Entropy_Partition}
    Let $\gamma: \mathcal{F} \rightarrow \mathbb{R}_{+}$ be any fractional partition with respect to a family $\mathcal{F}$ of subsets of $[1:n]$. For jointly distributed random variables $X_1, \ldots, X_n$, with $\mathsf{e}(F)=H(X_F)$, $F\subseteq[1:n]$ and any $\varepsilon \geq 0$, the following  holds:
    
    If $\emph{Gap}_{\emph{L}}(\mathsf{e}, \mathcal{F}, \gamma) \leq \varepsilon$ or $\emph{Gap}_{\emph{U}}(\mathsf{e}, \mathcal{F}, \gamma)\leq\varepsilon$, then 
    \begin{align}
        I(X_i; X_{[1:n]\setminus{\{i\}}})\leq\frac{\varepsilon}{\sigma}, \ \forall i \in [1:n],
    \end{align}
    where $\sigma = \min\limits_{\substack{i, j \in [1:n]: \\ i \neq j}} \sum\limits_{\substack{ F \in \mathcal{F}: \\ i \in F, j \notin F }} \gamma(F) > 0$.
\end{corollary}
The proof of Corollary~\ref{Stability_Entropy_Partition} is given in
\if \extended 1%
Appendix~\ref{appendix:e}.
\fi
\if \extended 0%
\cite[Appendix~E]{JakharKCP25}.
\fi
\begin{remark}
    The assertion of Corollary~\ref{Stability_Entropy_Partition} for $\text{Gap}_{\text{U}}$, in conjunction with Remark~\ref{remark:alphabeta}, recovers \cite[Lemma~9]{EllisKFY16} as a special case when $\alpha(F)=\frac{1}{k(\mathcal{F})}$, where $k(\mathcal{F})$ denotes the maximum integer $k$ such that every $i\in[1:n]$ belongs to at least $k$ members of $\mathcal{F}$. This general assertion, involving $\alpha$, is also implicitly hinted at in \cite[Prior to Theorem~4]{EllisKFY16}, albeit without a proof.
\end{remark}

\begin{corollary}\label{Equality_Implies_Independence}
    Let $\gamma: \mathcal{F} \rightarrow \mathbb{R}_{+}$, $\alpha: \mathcal{F} \rightarrow \mathbb{Q}_{+}$, and $\beta:\mathcal{F}\rightarrow \mathbb{Q}_+$ be any fractional partition, fractional covering, and fractional packing with respect to a family $\mathcal{F}$ of subsets of $[1:n]$. For jointly distributed random variables $X_1, \ldots, X_n$, with $\mathsf{e}(F)=H(X_F)$, $F\subseteq[1:n]$, the following hold:
    \begin{enumerate}[leftmargin=*]
        \item $\emph{Gap}_{\emph{U}}(\mathsf{e}, \mathcal{F}, \gamma)= 0$ (similarly, $\emph{Gap}_{\emph{L}}(\mathsf{e},\mathcal{F},\gamma)=0$) if and only if $X_i$, $i\in[1:n]$ are mutually independent. 
        \item $\emph{Gap}_{\emph{U}}(\mathsf{e}, \mathcal{F}, \alpha)=0$ (resp. $\emph{Gap}_{\emph{L}}(\mathsf{e}, \mathcal{F}, \beta)=0$) if and only if  $X_i$ for $i$ such that $\sum\limits_{F\in\mathcal{F}:i\in F}\alpha(F)=1$ (resp. $\sum\limits_{F\in\mathcal{F}:i\in F}\beta(F)=1$) are mutually independent, and $X_i$ for $i$ such that $\sum\limits_{F\in\mathcal{F}:i\in F}\alpha(F)>1$ (resp. $\sum\limits_{F\in\mathcal{F}:i\in F}\beta(F)<1$) are constants.
    \end{enumerate}
\end{corollary}

\begin{remark}
(i) The assertion of Corollary~\ref{Equality_Implies_Independence} on $\text{Gap}_{\text{L}}(\mathsf{e},\mathcal{F},\gamma)$ provides equality conditions for the inequality stating that the joint entropy upper bounds the erasure entropy~\cite[Theorem~1]{VerduW08} as a special case when $\mathcal{F} = \{\!\!\{ \{i\} : i \in [1:n]\}\!\!\}$ and $\gamma(F)=1$, $\forall F\in\mathcal{F}$. 

\noindent (ii) The assertion on $\text{Gap}_{\text{U}}(\mathsf{e}, \mathcal{F}, \alpha)$ recovers the equality conditions for Shearer's lemma~\cite{Radhakrishnan03} as a special case when $\alpha(F)=\frac{1}{k(\mathcal{F})}$, where $k(\mathcal{F})$ denotes the maximum integer $k$ such that each $i\in[1:n]$ belongs to at least $k$ members of $\mathcal{F}$. To the best of our knowledge, these conditions do not appear explicitly in the literature. While the independence of the random variables can be inferred from \cite[Proof of Lemma~9]{EllisKFY16}, the assertion that some of the random variables are constants does not follow from there.
\end{remark}

A detailed proof of Corollary~\ref{Equality_Implies_Independence} is given in
\if \extended 1%
Appendix~\ref{appendix:f}.
\fi
\if \extended 0%
\cite[Appendix~F]{JakharKCP25}.
\fi
While Corollaries~\ref{Stability_Entropy_Partition} and \ref{Equality_Implies_Independence} are stated for discrete entropy, we note that \eqref{eqn:MTIneq}, Corollary~\ref{Stability_Entropy_Partition}, and part 1) of Corollary~\ref{Equality_Implies_Independence} also hold for differential entropy, as it satisfies the submodularity property~\cite{Fujishige05}. However, part 2) of Corollary~\ref{Equality_Implies_Independence} does not hold because differential entropy is not generally non-decreasing. 

\subsection{Implications for Relative Entropy}
Let $P_{X_{[1:n]}}$ be any joint probability distribution and $Q_{X_{[1:n]}}$ be a product probability distribution on $\mathcal{X}^n$, i.e., $Q_{X_{[1:n]}}(x_{[1:n]})=\prod_{i=1}^nQ_{X_i}(x_i)$. For $F\subseteq [1:n]$, let $d(F)=-D(P_{X_{F}}\| Q_{X_F})$,
where $D(P_{X_F}\| Q_{X_F})$ denotes the relative entropy between the probability distributions $P_{X_F}$ and $Q_{X_F}$. In~\cite[Theorem~V]{MadimanT10}, it is shown that $d$ is submodular.

\begin{corollary} \label{Relative_Entropy_Partition_Equality}
    Let $P_{X_{[1:n]}}$ be any joint probability distribution and $Q_{X_{[1:n]}}$ be a product probability distribution on $\mathcal{X}^n$, and  $d(F)=-D(P_{X_{F}}\| Q_{X_F})$, $F\subseteq [1:n]$. Let $\gamma:\mathcal{F}\rightarrow \mathbb{R}_+$ and $\alpha:\mathcal{F}\rightarrow \mathbb{Q}_+$ be any fractional partition and fractional covering with respect to a family $\mathcal{F}$ of subsets of $[1:n]$. Then the following hold:
    \begin{enumerate}[leftmargin=*]
        \item $\emph{Gap}_{\emph{U}}(d,\mathcal{F},\gamma)=0$ if and only if $P_{X_{[1:n]}}$ is a product probability distribution.
        \item $\emph{Gap}_{\emph{U}}(d,\mathcal{F},\alpha)=0$ if and only if $P_{X_{[1:n]}}$ is a product probability distribution and $P_{X_{Z}}=Q_{X_{Z}}$, for $Z=\{i\in[1:n]:\sum\limits_{F\in\mathcal{F}:i\in F}\alpha(F)>1\}$.
        \end{enumerate}
\end{corollary}
The proof of Corollary~\ref{Relative_Entropy_Partition_Equality} is given in
\if \extended 1%
Appendix~\ref{appendix:g}.
\fi
\if \extended 0%
\cite[Appendix~G]{JakharKCP25}.
\fi

\subsection{Implications for Matroid Rank Function}
\begin{definition}[Matroid and Rank Function~\cite{Fujishige05}]
    A set system $(\mathcal{E}, \mathcal{I})$ where $\mathcal{I} \subseteq 2^{\mathcal{E}}$ is a matroid if
    \begin{enumerate}[leftmargin=*]
        \item $\phi \in \mathcal{I}$.
        \item $\forall I_2 \in \mathcal{I},~I_1 \subseteq I_2 \implies I_1 \in \mathcal{I}$.
        \item $\forall I_1, I_2 \in \mathcal{I}$, with $|I_1| < |I_2|$, there exists $e \in I_2 \setminus{I_1}$ such that $I_1 \cup \{e\} \in \mathcal{I}$.
    \end{enumerate}
    Given a matroid $\mathcal{M} = (\mathcal{E}, \mathcal{I})$, the rank function $r: 2^{\mathcal{E}} \rightarrow \mathbb{Z}_{\geq 0}$ is defined as $r(S) = \max_{I \in \mathcal{I}: I \subseteq S} |I|$, for all $S \subseteq \mathcal{E}$.
\end{definition}
Some examples of a matroid are as follows:

\begin{enumerate}[label=(\roman*)]
    \item $\mathcal{M} = (\mathcal{E}, \mathcal{I})$, where $\mathcal{E}$ is the set of column vectors corresponding to a matrix $\mathcal{A}$ and $\mathcal{I}$ is the set of those subsets of $\mathcal{E}$ which correspond to linearly independent column vectors in $\mathcal{A}$.
    \item $\mathcal{M} = (\mathcal{E}, \mathcal{I})$, where $\mathcal{E}$ is the set of edges of a graph $\mathcal{G} = (\mathcal{V}, \mathcal{E})$ on the set of vertices $\mathcal{V}$, and $\mathcal{I}$ is the set of subsets of edges which do not form a cycle in $\mathcal{G}$.
\end{enumerate}
As the rank function of a matroid is submodular~\cite{Fujishige05}, we have the following.
\begin{corollary} \label{Matroid_Rank_Equality}
    Let $\mathcal{M} = (\mathcal{E}, \mathcal{I})$ be any matroid with the rank function $r: 2^{\mathcal{E}} \rightarrow \mathbb{Z}_{\geq 0}$. Let $\gamma: \mathcal{F} \rightarrow \mathbb{R}_{+}$ be any fractional partition with respect to any family $\mathcal{F}$ of subsets of $\mathcal{E}$. Then, 
    \begin{align}
        \sum\limits_{F \in \mathcal{F}} \gamma(F) r(F) = r(\mathcal{E})
    \end{align}
    if and only if $\mathcal{I} = 2^A$, for some $A \subseteq \mathcal{E}$.
\end{corollary}

A notable example of a matroid for which Corollary~\ref{Matroid_Rank_Equality} holds is the free matroid, defined as $\mathcal{M}=(\mathcal{E},2^{\mathcal{E}})$~\cite{Fujishige05}. The proof of Corollary~\ref{Matroid_Rank_Equality} is given in 
\if \extended 1%
Appendix~\ref{appendix:h}.
\fi
\if \extended 0%
\cite[Appendix~H]{JakharKCP25}.
\fi

\section{Applications}
\subsection{A New Multivariate Mutual Information}
We propose a new multivariate mutual information, which is particularly motivated by part 1) of Corollary~\ref{Equality_Implies_Independence}.
\begin{definition}[{$(\mathcal{F},\gamma)$-Mutual Information}]
    Given a family $\mathcal{F}$ of subsets of $[1:n]$, a fractional partition $\gamma: \mathcal{F} \rightarrow \mathbb{R}_{+}$, and jointly distributed random variables $X_1, X_2, \ldots, X_n$, the $(\mathcal{F}, \gamma)$-mutual information (denoted by $(\mathcal{F}, \gamma)\text{-}MI$) of $X_1, X_2, \ldots, X_n$ is defined as
    \begin{align}
        (\mathcal{F}, \gamma)\text{-}MI(X_1; \cdots; X_n) & = \emph{Gap}_{\emph{U}}(\mathsf{e}, \mathcal{F}, \gamma) \nonumber \\
        & = \sum\limits_{F \in \mathcal{F}} \gamma(F)H(X_F) - H(X_{[1:n]}). \label{FGMI}
    \end{align}
\end{definition}
It follows from Theorems \ref{MT inequality} and Corollary \ref{Equality_Implies_Independence} that the $(\mathcal{F}, \gamma)\text{-}MI$ satisfies the usual criteria expected of a mutual information - namely, \emph{non-negativity} and \emph{independence property}. In addition, $(\mathcal{F}, \gamma)\text{-}MI(X_1; \cdots; X_n)$ recovers total correlation~\cite{Watanabe60} $\text{TC}(X_1; \cdots; X_n) = \sum_{i = 1}^n H(X_i) - H(X_{[1:n]})$, dual total correlation~\cite{Han75} $\text{DTC}(X_1; \cdots; X_n) = H(X_{[1:n]}) - \sum_{i = 1}^n H(X_i | X_{[1:n] \setminus{\{i\}}})$, and shared information~\cite{CsiszarN04,narayan2016multiterminal} $\text{SI}(X_1;\cdots;X_n)=H(X_{[1:n]})-\max\limits_{\gamma:\mathcal{B}\rightarrow \mathbb{R}_+}\sum\limits_{F\in\mathcal{B}}\gamma(F)H(X_F| X_{F^\text{c}})$, where $\mathcal{B}=2^{[1:n]}\setminus\{\phi,[1:n]\}$ and $\gamma$ denotes a fractional partition with respect to $\mathcal{B}$, as special cases. These are explicitly stated in the following proposition.
\begin{proposition}\label{prop:MMI-specialcases}
    \begin{enumerate}[leftmargin=*]
        \item $(\mathcal{F}, \gamma)\text{-}MI(X_1; \cdots; X_n) \geq 0$ with equality if and only if $X_1, \ldots, X_n$ are mutually independent.
        \item If $\mathcal{F} = \{\!\!\{\{i\}: i \in [1:n]\}\!\!\}$ and $\gamma(F) = 1, ~\forall F \in \mathcal{F}$, then $(\mathcal{F}, \gamma)\text{-}MI(X_1; \cdots; X_n) = \emph{TC}(X_1; \cdots; X_n)$.
        \item If $\mathcal{F} = \{\!\!\{ [1:n] \setminus{\{i\}} : i \in [1:n] \}\!\!\}$ and $\gamma(F) = 1/(n-1), ~\forall F \in \mathcal{F}$, then $(\mathcal{F}, \gamma)\text{-}MI(X_1; \cdots; X_n) = \frac{1}{n-1} \emph{DTC}(X_1; \cdots; X_n)$.
        \item For $\mathcal{F}=2^{[1:n]}\setminus\{\phi,[1:n]\}$, 
        \begin{align}
            \hspace{-12pt}\min_{\gamma:\mathcal{F}\rightarrow \mathbb{R}_+}\frac{(\mathcal{F},\gamma)\text{-}MI(X_1;\cdots;X_n)}{w(\gamma)-1}=\emph{SI}(X_1;\cdots;X_n).
        \end{align}
    \end{enumerate}
\end{proposition}
A detailed proof of Proposition~\ref{prop:MMI-specialcases} is given in 
\if \extended 1%
Appendix~\ref{appendix:i}.
\fi
\if \extended 0%
\cite[Appendix~I]{JakharKCP25}.
\fi
The reader might wonder why only the upper gap, $\text{Gap}_{\text{U}}$, is considered as the mutual information measure, rather than the lower gap, $\text{Gap}_{\text{L}}$. Due to the one-one correspondence between the upper and lower gaps, achieved through appropriately chosen $\mathcal{F}$ and $\gamma$ via duality (Theorem \ref{Duality}), it suffices to focus on either one of these gaps. We now present some properties of $(\mathcal{F}, \gamma)\text{-}MI$.
\begin{theorem} \label{Multivariate_Mutual_Information}
    \begin{enumerate}[leftmargin=*]
        \item $\max\limits_{\mathcal{F}, \gamma} (\mathcal{F}, \gamma)\text{-}MI(X_1; \cdots; X_n)\\ = \emph{TC}(X_1; \cdots; X_n)$.
        \item For any random variables $X_1, \ldots, X_n$ and $Y_1, \ldots, Y_n$, we have
        \begin{align}
            &(\mathcal{F}, \gamma)\text{-}MI(Y_1; \cdots; Y_n) \nonumber \\
            & \leq (\mathcal{F}, \gamma)\text{-}MI(X_1; \cdots; X_n) + \sum\limits_{i = 1}^n H(Y_i | X_i).
        \end{align}
        \item Given a family $\mathcal{F}$ and a fractional partition $\gamma:\mathcal{F}\rightarrow \mathbb{R}_+$, let $\tilde{\mathcal{F}} = \{\!\!\{ F \cap [1:n-1]: F \in \mathcal{F}\}\!\!\}$, and for each $\tilde{F} \in \tilde{\mathcal{F}}$, let $\tilde{\gamma}(\tilde{F}) = \gamma(F)$, where $F \in \mathcal{F}$  is the set corresponding to $\tilde{F}$. Then the following hold.
        \begin{align}
             &(\mathcal{F}, \gamma) \text{-}MI(X_1; \cdots; X_n) \nonumber \\
             & = (\mathcal{\tilde{F}}, \tilde{\gamma})\text{-}MI(X_1; \cdots; X_{n-1}) \nonumber \\
             & \hspace{12pt}+ \sum\limits_{\substack{F \in \mathcal{F}: \\ n \in F}} \gamma(F) I(X_n; X_{[1:n-1] \cap F^{\emph{c}}} | X_{[1:n-1] \cap F}),
             \label{eq: MMI_2c}
        \end{align}
        \begin{align}
            &(\mathcal{\tilde{F}}, \tilde{\gamma}) \text{-}MI(X_1; \cdots; X_{n-1}) \nonumber \\
            & \leq (\mathcal{F}, \gamma)\text{-}MI(X_1; \cdots; X_n) \nonumber \\ 
            & \leq (\mathcal{\tilde{F}}, \tilde{\gamma})\text{-}MI(X_1; \cdots; X_{n-1}) + I(X_n; X_{[1:n-1]}).
        \end{align}
        \item $(\mathcal{F}, \gamma)\text{-}MI(X_1; \cdots; X_n)$ is symmetric (i.e., invariant for every permutation among $X_1,\dots,X_n$) if and only if it is of the form $\sum\limits_{i = 1}^{n-1} \gamma_i \left( \sum\limits_{|F| = i} H(X_F) \right) - H(X_{[1:n]})$, where $\sum\limits_{i = 1}^{n-1} \gamma_i \binom{n-1}{i-1} = 1$.
    \end{enumerate}
\end{theorem}

\begin{remark}
    Property 1) in Theorem~\ref{Multivariate_Mutual_Information} generalizes Watanabe's observation~\cite{Watanabe60}, which states that $\max\limits_{\mathcal{P}} \sum\limits_{P \in \mathcal{P}} H(X_P) - H(X_{[1:n]}) = \text{TC}(X_1; \cdots; X_n)$, where $\mathcal{P}$ denotes a partition of $[1:n]$, to fractional partitions defined via $\mathcal{F}$ and $\gamma$. Special cases of properties 2) and 3) for TC and DTC appear in~\cite[Lemmas~4.5,~4.8~and~4.9]{Austin20},\cite[Equation~(13)]{AyOBJ06}.
\end{remark}

 Property 2) states that the mutual information among $Y_1,\dots,Y_n$ cannot exceed the mutual information among $X_1,\dots,X_n$, plus the equivocation of $Y_i$ given $X_i$ across all $i\in[1:n]$.  In particular, if $X_i$ almost surely determines $Y_i$ for each $i \in [1:n]$, then $(\mathcal{F}, \gamma) \text{-}MI(Y_1; \cdots; Y_n) \leq (\mathcal{F}, \gamma)\text{-}MI(X_1; \cdots; X_n)$. Property 3) provides a recursive formula for $(\mathcal{F}, \gamma)\text{-}MI(X_1; \cdots; X_n)$ in terms of that of fewer random variables with appropriately defined $\tilde{\mathcal{F}}$ and $\tilde{\gamma}$, and also shows that $(\mathcal{F}, \gamma)\text{-}MI(X_1; \cdots; X_n)$ is non-decreasing in the number of random variables. A detailed proof of Theorem~\ref{Multivariate_Mutual_Information} is given in 
\if \extended 1%
Appendix~\ref{appendix:j}.
\fi
\if \extended 0%
\cite[Appendix~J]{JakharKCP25}.
\fi

\subsection{Matrix Determinantal Inequalities}
Using information-theoretic inequalities to prove matrix determinantal inequalities for positive semidefinite matrices has been well-studied in the literature \cite{DemboCT91}. The following proposition presents the equality conditions for the determinantal inequalities proved in \cite[Corollary~III]{MadimanT10} using the fractional subadditivity of differential entropy.

\begin{proposition} \label{Determinantal_Equality}
    Let $K$ be a positive definite $n \times n$ matrix, and let $\mathcal{F}$ be a family of subsets on $[1:n]$. For $F \in \mathcal{F}$, let $K(F)$ denote the submatrix of $K$ corresponding to the rows and columns indexed by the elements of $F$. Then, using $|M|$ to denote the determinant of a matrix $M$, we have that, for any fractional partition $\gamma$ with respect to $\mathcal{F}$, 
    \begin{align}
        \prod\limits_{F \in \mathcal{F}} |K(F)|^{\gamma(F)} = |K| 
    \end{align}
    if and only if $K$ is a diagonal matrix, i.e., $K_{ij} = 0$ for all  $i \neq j$.
\end{proposition}
\begin{remark}
    Proposition \ref{Determinantal_Equality} recovers the equality conditions of the classical determinantal inequalities of Hadamard, Sz\'asz, and Fischer~\cite{HornJ12,DemboCT91} by choosing $\mathcal{F} = \{\!\!\{ \{i\} : i \in [1:n]\}\!\!\}$ with $\gamma(F) = 1,~\forall F \in \mathcal{F}$; $\mathcal{F} = \{\!\!\{ [1:n] \setminus{\{i\}} : i \in [1:n]\}\!\!\}$ with $\gamma(F) = \frac{1}{n-1},~\forall F \in \mathcal{F}$; $\mathcal{F} = \{\!\!\{F, F^{\text{c}}\}\!\!\}$, for any arbitrary $F \subset [1:n]$, with $\gamma(F) = \gamma(F^{\text{c}}) = 1$, respectively.
\end{remark}

The proof of Proposition~\ref{Determinantal_Equality} is given in
\if \extended 1%
Appendix~\ref{appendix:k}.
\fi
\if \extended 0%
\cite[Appendix~K]{JakharKCP25}.
\fi

\section{Acknowledgment}
    We would like to thank Vincent Tan for helpful discussions on connections between $(\mathcal{F},\gamma)$-$MI$ and shared information~\cite{CsiszarN04,narayan2016multiterminal}.

\IEEEtriggeratref{14}
\bibliographystyle{IEEEtran}
\bibliography{bibliofile}

\clearpage

\if \extended 1

\appendices
\section{Justification of Assumptions}\label{appendix:a}
When two indices always occur together in members of $\mathcal{F}$, they can be treated as a single index that embodies this pair of indices. Under this treatment, structure of $\mathcal{F}$ and the corresponding $\gamma$-values remain unchanged. This clarifies Assumption 1.

For Assumption 2, consider any family $\mathcal{F}$ and a fractional partition $\gamma: \mathcal{F} \rightarrow \mathbb{R}_{+}$ such that  $[1:n] \in \mathcal{F}$, and let $\delta=\gamma([1:n])$. Then from Theorem~\ref{MT inequality} we get that - for any submodular function $f: 2^{[1:n]} \rightarrow \mathbb{R}$ with $f(\phi) = 0$, we have 
\begin{align}
    f([1:n]) &\leq \sum\limits_{S \in \mathcal{F}} \gamma(S)f(S) \\
    &= \delta f([1:n]) + \sum\limits_{\substack{S \in \mathcal{F}: \\ i \in S, \\ S \neq [1:n]}} \gamma(S)f(S), 
\end{align}
which implies that
\begin{align}
    f([1:n])&\leq \frac{1}{1 - \delta} \sum\limits_{\substack{S \in \mathcal{F}: \\ i \in S, \\ S \neq [1:n]}} \gamma(S)f(S) \label{eq: As2}.
\end{align}
Let us now define a new family $\mathcal{F}'=\{\!\!\{S\in\mathcal{F}:S\neq [1:n]\}\!\!\}$ and a fractional partition $\gamma': \mathcal{F}' \rightarrow \mathbb{R}_{+}$ such that $\gamma'(S) = \frac{\gamma(S)}{1 - \delta}, ~\forall S \in \mathcal{F}'$. Note that it is a valid fractional partition since for each $i \in [1:n]$,
\begin{align}
    \sum\limits_{\substack{S \in \mathcal{F}': \\ i \in S}} \gamma'(S) & = \sum\limits_{\substack{S \in \mathcal{F}: \\ i \in S, \\ S \neq [1:n]}} \frac{\gamma(S)}{1 - \delta} \\
    & = \frac{1}{1 - \delta} \sum\limits_{\substack{S \in \mathcal{F}: \\ i \in S, \\ S \neq [1:n]}} \gamma(S) \\
    & = \frac{1}{1 - \delta} \left( \sum\limits_{\substack{S \in \mathcal{F}: \\ i \in S}} \gamma(S) - \delta \right) \label{eq: assum_2.1} \\
    & = \frac{1 - \delta}{1 - \delta} = 1, \label{eq: assum_2.2}
\end{align}
where \eqref{eq: assum_2.1} and \eqref{eq: assum_2.2} hold because $\gamma([1:n]) = \delta$ and $\sum\limits_{\substack{S \in \mathcal{F}: \\ i \in S}} \gamma(S) = 1$, respectively.
Then, \eqref{eq: As2} can be re-expressed as follows.
\begin{align}
    f([1:n]) \leq \sum\limits_{\substack{S \in \mathcal{F}': \\ i \in S}} \gamma'(S)f(S).
\end{align}
Thus, removing $[1:n]$ from the family $\mathcal{F}$ is permissible, as we can always readjust the fractional partition in a way that preserves the integrity of the inequality. Although we focused just on the upper bound in Theorem \ref{MT inequality}, the same reasoning also holds for the lower bound. Further, the arguments remain valid even when we consider fractional covering (or packing) instead of fractional partition.

For Assumption 3, consider any family $\mathcal{F}$ and a fractional partition $\gamma: \mathcal{F} \rightarrow \mathbb{R}_{+}$ such that $\exists S' \in \mathcal{F}$ such that $\gamma(S') = 0$. Here we argue that we can always remove that $S'$ from the family $\mathcal{F}$ to get $\mathcal{F}' = \mathcal{F} \setminus{S'}$ and define another $\gamma': \mathcal{F}' \rightarrow \mathbb{R}_{+}$ such that $\gamma'(S) = \gamma(S), ~\forall S \in \mathcal{F}'$, and the inequality concerning the fractional subadditivity remains unchanged.

To restate, the arguments for Assumption 3 go through for the lower bound in Theorem \ref{MT inequality}, and for any fractional covering (or packing).

\section{Proof of Theorem \ref{Stability_Submodular_Function_Partition}}\label{appendix:b}

We first derive a lower bound for $\text{Gap}_\text{U}(f, \mathcal{F}, \gamma)$ as follows.
\begin{align}
    &\text{Gap}_\text{U}(f, \mathcal{F}, \gamma) \nonumber \\
    &= \sum\limits_{S \in \mathcal{F}} \gamma(S)f(S) - f([1:n]) \\
    &= \sum\limits_{\substack{S \in \mathcal{F}: \\ n \in S}} \gamma(S)f(S) + \sum\limits_{\substack{S \in \mathcal{F}: \\ n \notin S}} \gamma(S)f(S) - f([1:n]) \\
    &= \sum\limits_{\substack{S \in \mathcal{F}: \\ n \in S}} \gamma(S) \left[ f(\{n\}) + f({S \setminus{\{n\}}|\{n\}}) \right] \nonumber \\
    & \quad \quad + \sum\limits_{\substack{S \in \mathcal{F}: \\ n \notin s}} \gamma(S)f(S) - f(\{n\}) - f([1:n-1]|\{n\}) \\
    &= \sum\limits_{\substack{S \in \mathcal{F}: \\ n \in S}} \gamma(S) f({S \setminus{\{n\}}|\{n\}}) + \sum\limits_{\substack{S \in \mathcal{F}: \\ n \notin S}} \gamma(S)f(S) \nonumber \\
    & \quad \quad - f([1:n-1]|\{n\}) \label{eq: SSFP_1.1}\\
    &= \sum\limits_{\substack{S \in \mathcal{F}: \\ n \in S}} \gamma(S) \sum\limits_{j \in S \setminus{\{n\}}} f(\{j\}|[1:j-1] \cap (S \setminus{\{n\}}), \{n\}) \nonumber \\
    & \quad \quad + \sum\limits_{\substack{ S \in \mathcal{F}: \\ n \notin S}} \gamma(S) \sum\limits_{j \in S} f(\{j\} | [1:j-1] \cap S) \nonumber \\
    & \quad \quad - f([1:n-1]|\{n\}) \label{eq: SSFP_1.2} \\
    &= \sum\limits_{j = 1}^{n - 1} \sum\limits_{\substack{S \in \mathcal{F}: \\ n \in S, \\ j \in S}} \gamma(S) f(\{j\} | [1:j-1] \cap (S \setminus{\{n\}}), \{n\}) \nonumber \\
    & \quad \quad + \sum\limits_{j = 1}^{n - 1} \sum\limits_{\substack{S \in \mathcal{F}: \\ n \notin S, \\ j \in S}} \gamma(S) f(\{j\} |[1:j-1] \cap S) \nonumber \\
    & \quad \quad - f([1:n-1]|\{n\}) \label{eq: SSFP_1.3}\\
    &\geq \sum\limits_{j = 1}^{n - 1} \sum\limits_{\substack{S \in \mathcal{F}: \\ n \in S, \\ j \in S}} \gamma(S) f(\{j\} | [1:j-1], \{n\}) \nonumber \\
    &\quad \quad + \sum\limits_{j = 1}^{n - 1} \sum\limits_{\substack{S \in \mathcal{F}: \\ n \notin S, \\ j \in S}} \gamma(S) f(\{j\} | [1:j-1]) \nonumber \\
    & \quad \quad - f([1:n-1]|\{n\}) \label{eq: SSFP_1.4} \\
    &= \sum\limits_{j = 1}^{n - 1} f(\{j\} | [1:j-1], \{n\}) \left( 1 - \sum\limits_{\substack{S \in \mathcal{F}: \\ n \notin S, \\ j \in S}} \gamma(S) \right) \nonumber \\
    &\quad \quad + \sum\limits_{j = 1}^{n - 1} f(\{j\} | [1:j-1]) \left( \sum\limits_{\substack{S \in \mathcal{F}: \\ n \notin S, \\ j \in S}} \gamma(S) \right) \nonumber \\
    & \quad \quad - f([1:n-1]|\{n\}) \label{eq: SSFP_1.8}\\
    & \geq \sum\limits_{j = 1}^{n - 1} f(\{j\} |[1:j-1], \{n\}) \nonumber \\
    & +  \sum\limits_{j = 1}^{n - 1} \sigma\left( f(\{j\} | [1:j-1]) - f(\{j\} | [1:j-1], \{n\}) \right) \nonumber \\
    & \quad \quad - f([1:n-1]|\{n\}) \label{eq: SSFP_1.5} \\
    & = \sigma \sum\limits_{j = 1}^{n - 1} f(\{j\} | [1:j-1]) \nonumber \\
    & \quad \quad - \sigma \sum\limits_{j = 1}^{n - 1} f(\{j\} | [1:j-1], \{n\}) \label{eq: SSFP_1.9}\\
    & = \sigma \left[ f([1:n-1]) - f([1:n] |\{n\}) \right] \label{eq: SSFP_1.10}\\
    &= \sigma \left[ f([1:n-1]) + f(\{n\}) - f([1:n]) \right]. \label{eq: SSFP_1.6}
\end{align}
In the above math block, \eqref{eq: SSFP_1.1} holds because $\sum\limits_{\substack{S \in \mathcal{F}: \\ n \in S}} \gamma(S) = 1$, \eqref{eq: SSFP_1.4} holds because $f(\{j\} | S_1, S_2) \leq f(\{j\} | S_1),\ \forall S_1, S_2 \subseteq [1:n]$, \eqref{eq: SSFP_1.3} follows by interchanging the summations, \eqref{eq: SSFP_1.8} follows from the fact that 
\begin{align}
    \sum\limits_{\substack{S \in \mathcal{F}: \\ n \in S, \\ j \in S}} \gamma(S) & = \sum\limits_{\substack{S \in \mathcal{F}: \\ j \in S}} \gamma(S) - \sum\limits_{\substack{S \in \mathcal{F}: \\ n \notin S, \\ j \in S}} \gamma(S) \\
    & = 1 - \sum\limits_{\substack{S \in \mathcal{F}: \\ n \notin S, \\ j \in S}} \gamma(S),
\end{align}
\eqref{eq: SSFP_1.5} follows from the definition of $\sigma = \min\limits_{\substack{i, j \in [1:n]: \\ i \neq j}} \sum\limits_{\substack{ S \in \mathcal{F}: \\ i \in S, j \notin S }} \gamma(S)$, and \eqref{eq: SSFP_1.2}, \eqref{eq: SSFP_1.9}, \eqref{eq: SSFP_1.10}, \eqref{eq: SSFP_1.6} all follow from the chain rule. By putting together \eqref{eq: SSFP_1.6} with the given condition that $\text{Gap}_\text{U}(f, \mathcal{F}, \gamma) \leq \varepsilon$, we get
\begin{align*}
    f(n) + f([1:n-1]) - f([1:n]) \leq \frac{\varepsilon}{\sigma}.
\end{align*}
Reworking the proof as described above with an arbitrary $i \in [1:n]$ instead of $n$ by considering an arbitrary permutation of $[1:n]$ that maps $n$ to $i$, we obtain the following.
\begin{align}
         f(\{i\}) + f([1:n] \setminus{\{i\}})  - f([1:n]) \leq \frac{\varepsilon}{\sigma}, \hfill \forall i \in [1:n]. \label{eq: SSFP_1.7}
\end{align}

We now argue that $\sigma>0$. It suffices to prove the following claim as $\gamma(S)>0$, $\forall S\in\mathcal{F}$ by our \emph{initial assumptions}.
    \begin{claim} \label{set_cont_i_but_not_j}
    For $i,j\in[1:n]$ such that $i\neq j$, $\exists S \in \mathcal{F}$ such that $ i \in S, j \notin S$.
    \end{claim}
    \begin{proof}[Proof of Claim \ref{set_cont_i_but_not_j}]
        Assume, for the sake of contradiction, that $\forall S \in \mathcal{F}$, $i \in S$ implies $j \in S$.
        
        Case (i): $i$ and $j$ always occur together \\
        Then, they should have been clubbed together, according to our \textit{initial assumptions}.
        
        Case (ii): $i$ always occurs with $j$, and $j$ occurs separately,
        i.e., $|\{\!\!\{S' \in \mathcal{F} : j \in S'\}\!\!\}| > |\{\!\!\{S' \in \mathcal{F} : i \in S'\}\!\!\}|$.
        
        Then
        \begin{align}
            \sum\limits_{\substack{S' \in \mathcal{F}: \\ j \in S'}} \gamma(S') & = \sum\limits_{\substack{S' \in \mathcal{F}: \\ i \in S'\\j \in S'}} \gamma(S') + \sum\limits_{\substack{S' \in \mathcal{F}: \\ j \in S' \\ i \notin S'}} \gamma(S')  \\ 
            & = 1 + \sum\limits_{\substack{S' \in \mathcal{F}: \\ j \in S' \\ i \notin S'}} \gamma(S')  \label{eq: SCIBNJ_0}\\
            &>1 \label{eq: SCIBNJ_1},
        \end{align}
        where \eqref{eq: SCIBNJ_0} follows from the definition of $\gamma$ and the fact that $i$ always appears with $j$, and \eqref{eq: SCIBNJ_1} follows because $j$ occurs separately and $\gamma(S') > 0, \forall S' \in \mathcal{F}$. 
        Thus, we arrive at a contradiction.
    \end{proof}
    
We now prove that \eqref{eq: SSFP_1.7} continues to hold when  $\textup{Gap}_{\textup{L}}(f, \mathcal{F}, \gamma)\leq \varepsilon$, instead of $\textup{Gap}_{\textup{U}}(f, \mathcal{F}, \gamma)\leq\varepsilon$. Using duality (Theorem~\ref{Duality}), we have
\begin{align}
    \textup{Gap}_{\textup{U}}(f,\bar{\mathcal{F}},\bar{\gamma})&=\frac{w(\bar{\gamma})}{w(\gamma)}\textup{Gap}_{\textup{L}}(f,{\mathcal{F}},{\gamma})\\
    &\leq \frac{\epsilon}{w(\gamma)-1}\label{SSFP_2.2},
\end{align}
where \eqref{SSFP_2.2} follows because $\textup{Gap}_{\textup{L}}(f, \mathcal{F}, \gamma)\leq \varepsilon$ and 
\begin{align}
    w(\bar{\gamma}) & = \sum\limits_{S^{\text{c}} \in \bar{\mathcal{F}}} \bar{\gamma}(S^\text{c})\\
    &= \sum\limits_{S^{\text{c}} \in \bar{\mathcal{F}}} \frac{\gamma(S)}{w({\gamma}) - 1} \\
    & = \frac{1}{w({\gamma}) - 1} \sum\limits_{S\in {\mathcal{F}}} {\gamma}(S)\\
    &= \frac{w({\gamma})}{w({\gamma}) - 1}.
\end{align}
Applying our conclusion for the upper bound (i.e., \eqref{eq: SSFP_1.7}) to \eqref{SSFP_2.2}, we get, $\forall i \in [1:n]$,

\begin{align}
         f(\{i\}) + f([1:n] \setminus{\{i\}})  - f([1:n]) &\leq \frac{\varepsilon}{\sigma'(w(\gamma)-1)}\\
         &=\frac{\varepsilon}{\sigma} \label{eq: neweqn},
\end{align}
 where \eqref{eq: neweqn} follows because

\begin{align*}
    \sigma' & = \min\limits_{\substack{i, j \in [1:n]: \\ i \neq j}} \sum\limits_{\substack{ S^{\text{c}} \in \bar{\mathcal{F}}: \\ i \in S^{\text{c}}, j \notin S^{\text{c}} }} \bar{\gamma}(S^{\text{c}}) \\
    & = \min\limits_{\substack{i, j \in [1:n]: \\ i \neq j}} \sum\limits_{\substack{ S \in {\mathcal{F}}: \\ i \notin S, j \in S }} \frac{{\gamma}(S)}{w({\gamma}) - 1} \\
    & = \frac{1}{w({\gamma}) - 1} \min\limits_{\substack{i, j \in [1:n]: \\ i \neq j}} \sum\limits_{\substack{ S \in {\mathcal{F}}: \\ i \in S, j \notin S }} {\gamma}(S) \\
    & = \frac{\sigma}{w({\gamma}) - 1}.
\end{align*}
This completes the proof of Theorem~\ref{Stability_Submodular_Function_Partition}.

\emph{Details Omitted from Remark~\ref{remark:alphabeta}}: Firstly, it is important to mention that the duality, while stated for $\gamma$ in Theorem~\ref{Duality}, also holds for $\alpha$ (fractional covering) and $\beta$ (fractional packing), as discussed in \cite[Discussion after Theorem 4]{MadimanT10}. Remark~\ref{remark:alphabeta} states that Theorem~\ref{Stability_Submodular_Function_Partition} can be extended to incorporate fractional covering and packing in place of fractional partition.
Formally, we get the following corollary.
\begin{corollary}
    Let $f: 2^{[1:n]} \rightarrow \mathbb{R}$ be any submodular function with $f(\phi) = 0$. Let $\alpha:\mathcal{F}\rightarrow \mathbb{R}_+$ be any fractional covering with respect to a family $\mathcal{F}$ of subsets of $[1:n]$ and $\bar{\alpha}$ be the dual fractional packing corresponding to $\alpha$. For any $\varepsilon > 0$, the following holds:
    
        If $\sigma > 0$, and $\emph{Gap}_{\emph{L}}(f, \mathcal{F}, \bar{\alpha}) \leq  \varepsilon$ or $\emph{Gap}_{\emph{U}}(f, \mathcal{F}, \alpha) \leq  \varepsilon$, then
        \begin{align}
         f(\{i\}) + f([1:n] \setminus{\{i\}})  - f([1:n]) \leq \frac{\varepsilon}{\sigma},\  \forall i \in [1:n],
        \end{align}
      where $\sigma = \min\limits_{\substack{i, j \in [1:n]: \\ i \neq j}} \sum\limits_{\substack{ S \in \mathcal{F}: \\ i \in S, j \notin S }} \alpha(S)$.
\end{corollary}
\emph{Proof}: This proof closely follows that of Theorem~\ref{Stability_Submodular_Function_Partition}, starting with a lower bound on $\text{Gap}_{\text{U}}(f, \mathcal{F}, \alpha)$. The only place where this proof differs from the earlier one is \eqref{eq: SSFP_1.1}. Instead of equality, we get an inequality as $\sum\limits_{{S \in \mathcal{F}: n \in S}} \alpha(S) \geq 1$ and $f$ is non-negative. Consequently, we arrive at the same conclusion \eqref{eq: SSFP_1.7} for $\text{Gap}_{\text{U}}(f, \mathcal{F}, \alpha)$ as well.

The proof for $\text{Gap}_{\text{L}}(f, \mathcal{F}, \bar{\alpha})$ follows exactly the same steps as the proof for $\text{Gap}_{\text{L}}(f, \mathcal{F}, \gamma)$ in Theorem~\ref{Stability_Submodular_Function_Partition}, using duality and the fact that the dual of a fractional packing is a fractional covering (Interested reader can refer to \cite[After Definition~VI]{MadimanT10} for details).

\section{Proof of Theorem \ref{Equality_Implies_Modularity_Partition}}\label{appendix:c}
\emph{Proof of Part 1).} Let $\text{Gap}_\text{U}(f, \mathcal{F}, \gamma)$ be equal to $0$. Substituting $\varepsilon = 0$ in Theorem~\ref{Stability_Submodular_Function_Partition} and noting that $\sigma > 0$, we get
\begin{align}
    f(\{i\}) + f([1:n] \setminus{\{i\}}) - f([1:n]) & = 0, ~\forall i \in [1:n],
\end{align}
which implies that
\begin{align}
    f(\{i\}) - f(\{i\} ~|~ [1:n] \setminus{\{i\}}) & = 0, ~\forall i \in [1:n], \label{eq: EIMP_1.1}
\end{align}
where $f(S|T)$ is as defined in Section II. Since conditioning reduces the values of a submodular function by \eqref{eqref: conditioning}, we get
\begin{align}
    f(\{i\}) & \geq f(\{i\} | [1:i-1]) \\
    & \geq f(\{i\} | [1:n] \setminus{\{i\}}).
\end{align}
Now, from \eqref{eq: EIMP_1.1}, we get
\begin{align}
    f(\{i\}) = f(\{i\} | [1:i-1]). \label{eq: EIMP_1.2}
\end{align}
Using the chain rule \eqref{eqref: chain_rule} for $f$, we have
\begin{align}
    f([1:n]) &= \sum\limits_{i = 1}^{n} f(\{i\} ~|~ [1:i-1]) \\
    &= \sum\limits_{i = 1}^{n} f(\{i\}), \label{eq: EIMP_1.3}
\end{align}
where deduction to \eqref{eq: EIMP_1.3} uses \eqref{eq: EIMP_1.2}. Fix an arbitrary $S \subseteq [1:n]$. From \eqref{eq: EIMP_1.3}, we have
\begin{align}
    0 &= f([1:n]) - \sum\limits_{i = 1}^{n} f(\{i\}) \\
    &= f(S) + f(S^{\text{c}}|S) - \sum\limits_{i \in S} f(\{i\}) - \sum\limits_{i \in S^{\text{c}}} f(\{i\}) \label{eq: EIMP_1.4} \\
   &= \underbrace{\Big(f(S) - \sum\limits_{i \in S} f(\{i\})}_{\leq 0}\Big) + \underbrace{\Big(f(S^{\text{c}} | S) - \sum\limits_{i \in S^{\text{c}}} f(\{i\})\Big)}_{\leq 0}, \label{eq: EIMP_1.5}
\end{align}
where \eqref{eq: EIMP_1.4} follows from \eqref{eqref: chain_rule}. Notice that both the expressions in \eqref{eq: EIMP_1.5} are non-positive because $f(S) \leq \sum\limits_{i \in S} f(\{i\})$ by submodularity, and $f(S^{\text{c}} | S) \leq f(S^{\text{c}}) \leq \sum\limits_{i \in S^{\text{c}}} f(\{i\})$ by \eqref{eqref: conditioning} and submodularity. Since the sum of these two expressions is zero, each of them must be equal to zero. In particular, $f(S) = \sum\limits_{i \in S} f(\{i\})$. As $S \subseteq [1:n]$ is arbitrary, this proves that  
$f$ is modular.

For the other direction, let us suppose that $f$ is modular. Then,
\begin{align}
     \textup{Gap}_{\textup{U}}(f, \mathcal{F}, \gamma) & = \sum\limits_{S \in \mathcal{F}} \gamma(S) f(S)  - f([1:n]) \\
    & = \sum\limits_{S \in \mathcal{F}} \gamma(S) \sum\limits_{i \in S} f(\{i\}) - f([1:n]) \label{eq: EIMP_1.9}\\
    & = \sum\limits_{i = 1}^n f(\{i\})\left( \sum\limits_{\substack{S \in \mathcal{F}: \\ i \in S}} \gamma(S) \right) - f([1:n]) \\
    & = \sum\limits_{i = 1}^n f(\{i\}) - f([1:n]) \label{eq: EIMP_1.7} \\
    & = 0, \label{eq: EIMP_1.8}
\end{align}
where \eqref{eq: EIMP_1.9} and \eqref{eq: EIMP_1.8} use the fact that $f$ is modular and \eqref{eq: EIMP_1.7} follows because $\sum\limits_{\substack{S \in \mathcal{F}: \\ i \in S}} \gamma(s) = 1$. This proves that $\text{Gap}_\text{U}(f, \mathcal{F}, \gamma) = 0$ if and only if $f$ is modular. The assertion for $\text{Gap}_\text{L}(f, \mathcal{F}, \gamma)$ can be proved using this together with duality (Theorem~\ref{Duality}), as in the proof of Theorem~\ref{Stability_Submodular_Function_Partition}.

\emph{Proof of Part 2).} We first state and prove a claim that is essential for our proof of this part.  
\begin{claim}\label{Shearer_Lemma_Equality_Submodular}
    Let $f: 2^{[1:n]} \rightarrow \mathbb{R}$ be any submodular function with $f(\phi) = 0$ such that $f$ is non-decreasing, i.e., $f(S)\leq f(T)$ for $S\subseteq T$. Let $\mathcal{F}$ be any family of subsets of $[1:n]$ such that $\forall i \in [1:n], ~|\{S \in \mathcal{F} : i\in S\}| \geq k$, then 
    \begin{align*}
        k f([1:n]) = \sum\limits_{S \in \mathcal{F}} f(S) 
    \end{align*}
    if and only if $f$ is modular and $f(Z) = 0, ~\forall Z \subseteq \{ i \in [1:n] : |\{S \in \mathcal{F} : i\in S\}| > k \}$
\end{claim}
\begin{proof}[Proof of Claim~\ref{Shearer_Lemma_Equality_Submodular}]
Let $\mathcal{F} =\{\!\!\{S_1, \ldots, S_r\}\!\!\}$. Construct a new family $\mathcal{F'} =\{\!\!\{S_1', \ldots, S_r'\}\!\!\}$ of subsets of $[1:n]$ such that $S_j^\prime\subseteq S_j$, $\forall j\in[1:r]$ and each $i\in[1:n]$ appears exactly in $k$ members of $\mathcal{F}'$ (notice that there are many possible ways of constructing such an $\mathcal{F}'$). Now   $k f([1:n]) = \sum\limits_{S \in \mathcal{F}} f(S)$ implies that
\begin{align}
0&=\sum_{S \in \mathcal{F}} f(S)-kf([1:n])\\
&=\underbrace{\left(\sum_{S' \in \mathcal{F}'} f(S') -kf([1:n])\right)}_{\geq 0}+ \underbrace{\sum\limits_{j \in [1:r]} f((S_j \setminus{S_j'}) | S_j')}_{\geq 0} \label{eq: nSLES_1},
\end{align}
where \eqref{eq: nSLES_1} follows from the chain rule \eqref{eqref: chain_rule}. Notice that both the expressions in \eqref{eq: nSLES_1} are non-negative. The non-negativity of the first expression follows Theorem~\ref{MT inequality}, while the non-negativity of the second expression follows from the fact that $f(S_j) \geq f(S_j')$, as $f$ is non-decreasing. Since the sum of these two expressions is zero, each of them must be equal to zero. 

Since, there can be many possible $\mathcal{F}'$s depending on which sets we choose for deleting extra appearances of the indices that appear more than $k$ times, it can be the case that such an $\mathcal{F}'$ does not satisfy assumption 1 (stated in the beginning of Section~\ref{section:stab-eql}), i.e., some indices might always appear together in $\mathcal{F}'$. Based on this, we have the following cases. Note that $0 = f(\phi) \leq f(\{i\}),~\forall i \in [1:n]$ by the non-decreasing nature of $f$.

\textit{Case 1: $\mathcal{F}'$ is such that no two indices $i,j\in[1:n]$ always appear together in the members of $\mathcal{F}'$}. 

For this case, equality to zero for the first expression in \eqref{eq: nSLES_1} implies that $f$ is modular, by part 1) of Theorem~\ref{Equality_Implies_Modularity_Partition}. The second expression being equal to zero implies that, for all $j\in[1:r]$,
\begin{align}
    0 & = f((S_j \setminus{S_j'}) | S_j') \\
    & = f(S_j) - f(S_j') \label{eq;nneweqn1}\\
    &= f(S_j \setminus{S_j'})\label{eq: nSLES_2} \\
    & = \sum\limits_{i \in S_j \setminus{S_j'}} f(\{i\}), \label{eq: nSLES_3} 
\end{align}
where \eqref{eq;nneweqn1} follows from the chain rule \eqref{eqref: chain_rule}, and \eqref{eq: nSLES_2} and \eqref{eq: nSLES_3} follow from $f$ being modular. From this we get that
\begin{align}
 f(\{i\})=0, \label{eq: nSLES_4}
 \end{align}
 $\forall i \in \bigcup_{j \in [1:r]} (S_j \setminus S_j')=\{ i \in [1:n]  : |\{S \in \mathcal{F} : i\in S\}| > k \}$. Thus, by modularity of $f$, we have $f(Z) = 0, \forall Z \subseteq \{ i \in [1:n] : |\{S \in \mathcal{F} : i\in S\}| > k \}$.

\textit{Case 2:  There are some indices which appear always together in $\mathcal{F}'$}.

Let $G_1, \ldots, G_m$ be subsets of $[1:n]$ such that for all $\ell \in [1:m]$, elements of $G_{\ell}$ always appear together in $\mathcal{F}'$. It is important to note that $|\{ a_{\ell,v} \in G_{\ell}: |\{ S \in \mathcal{F}: a_{\ell,v} \in S\}| = k \}| \leq 1,~\forall \ell \in [1:m]$, i.e., in any of the groups $G_{\ell}, \ell \in [1:m]$, there can be at most one index  which has exactly $k$ appearances in $\mathcal{F}$. To see this, assume that there are two indices in $G_{\ell}$ which appear exactly $k$ times in $\mathcal{F}$. Let $a_1$ and $a_2$ be those indices without loss of generality. As none of them have been deleted in the construction of $\mathcal{F}'$, this means that both of them also appear together in $\mathcal{F}$. This leads us to a contradiction as no two indices in $[1:n]$ always appear together in $\mathcal{F}$.

Let us now analyze the conditions under which both the expressions in \eqref{eq: nSLES_1} are equal to zero. Let $G = [1:n] \setminus{\bigcup_{\ell = 1}^m G_{\ell}}$. By treating each group in $\mathcal{F}'$ as a single index (as per assumption 1), part 1) of Theorem~\ref{Equality_Implies_Modularity_Partition} implies that
\begin{align}
    f\bigg(\bigcup_{\ell \in \mathcal{T}} G_{\ell}, H \bigg) = \sum\limits_{i \in H} f(\{i\}) + \sum\limits_{\ell \in \mathcal{T}} f(G_{\ell}),
\end{align}
($f(S,T)$ is used to denote $f(S \cup T)$) where $H \subseteq G, \mathcal{T} \subseteq [1:m]$, for the first expression to be zero. The second expression in \eqref{eq: nSLES_1} being equal to zero implies that $\forall j \in [1:r]$,
\begin{align}
    0 & = f((S_j \setminus{S_j'}) | S_j') \\
    & = f(S_j) - f(S_j'), \label{eq: nSLES_5}
\end{align}
where \eqref{eq: nSLES_5} follows from chain rule \eqref{eqref: chain_rule}.
We now make use of the following lemma and defer its proof to Appendix~\ref{appendix:l}.
\begin{lemma}\label{lemma1}
    For disjoint sets $G,G_1,\ldots,G_m$ such that $\bigcup_{\ell = 1}^m G_{\ell} \cup G = [1:n]$ and any submodular set function $g: 2^{[1:n]} \rightarrow R$ such that $\forall H \subseteq G, \forall \mathcal{T} \subseteq [1:m],~ g\big(\bigcup_{\ell \in \mathcal{T}} G_{\ell}, H \big) = \sum\limits_{i \in H} g(\{i\}) + \sum\limits_{\ell \in \mathcal{T}} g(G_{\ell})$, the following holds:
    \begin{align}
        g\bigg(\bigcup_{\ell = 1}^m S_{\ell}, S \bigg) = \sum\limits_{i \in S} g(\{i\}) + \sum\limits_{\ell = 1}^m g(S_{\ell}),
    \end{align}
    $\forall S \subseteq G, S_1 \subseteq G_1, \ldots, S_m \subseteq G_m$.
\end{lemma}
Noting that $S_j'$ in \eqref{eq: nSLES_5} can have some groups of indices from $G_{\ell}: \ell \in [1:m]$, let $\mathcal{I} \subseteq [1:m]$ be the set of indices of the groups which appear in $S_j'$, i.e., $S_j' \cap \bigcup_{\ell=1}^m G_{\ell} = \bigcup_{\ell \in \mathcal{I}} G_{\ell}$. Let $R$ denote the indices that appear exactly $k$ times in $\mathcal{F}$, i.e., $R = \{ i \in [1:n]: | \{ S \in \mathcal{F}: i \in S\}| = k\}$. We now invoke Lemma \ref{lemma1} to rewrite \eqref{eq: nSLES_5} observing that $S_j'$ can have indices from $G$ and all indices from $G_{\ell}: \ell \in \mathcal{I}$, while $S_j$ can have indices from $G$, all indices from $G_{\ell}: \ell \in \mathcal{I}$, and some subsets of indices that appear in groups in $\mathcal{F}'$ with strictly greater than $k$ appearances in $\mathcal{F}$ (because indices that appear in groups with exactly $k$ appearances can only appear in $G_{\ell}: \ell \in \mathcal{I}$). 
\begin{align}
    0 & = \sum\limits_{\substack{i \in S_j:\\ i \in G}} f(\{i\}) + \sum\limits_{\ell = 1}^m f(S_j \cap G_{\ell})\\
    & = \sum\limits_{\substack{i \in S_j:\\ i \in G}} f(\{i\}) + \sum_{\ell \in \mathcal{I}} f(G_{\ell}) + \sum\limits_{\substack{\ell' \in [1:m] \setminus{\mathcal{I}}:\\ (S_j \cap G_{\ell'}) \cap R = \phi}} f(S_j \cap G_{\ell'}) \nonumber \\
    & \quad\quad- \sum\limits_{\substack{i \in S_j':\\ i \in G}} f(\{i\}) - \sum_{\ell \in \mathcal{I}} f(G_{\ell}) \label{eq: nSLES_6} \\
    & = \sum\limits_{\substack{i \in S_j:\\ i \in G}} f(\{i\}) - \sum\limits_{\substack{i \in S_j':\\ i \in G}} f(\{i\}) + \sum\limits_{\substack{\ell' \in [1:m] \setminus{\mathcal{I}}:\\ (S_j \cap G_{\ell'}) \cap R = \phi}} f(S_j \cap G_{\ell'}) \nonumber \\
    & = \sum\limits_{\substack{i \in S_j \setminus{S_j'}:\\ i \in G}} f(\{i\}) + \sum\limits_{\substack{\ell' \in [1:m] \setminus{\mathcal{I}}:\\ (S_j \cap G_{\ell'}) \cap R = \phi}} f(S_j \cap G_{\ell'})\\
    & = \sum\limits_{\substack{i \in S_j \setminus{S_j'}:\\ i \in G}} f(\{i\}) + \sum\limits_{\ell = 1}^m f( (S_j \setminus{S_j'}) \cap G_{\ell}) \\
    & = f(S_j \setminus{S_j'}), \label{eq: nSLES_7}
\end{align}
where \eqref{eq: nSLES_7} follows because the only indices in $S_j \setminus{S_j'}$ are the indices which appear strictly greater than $k$ times in $\mathcal{F}$ but may or may not appear in groups in $\mathcal{F}'$. From this we get that,
\begin{align}
    f(\{i\}) = 0,
\end{align}
for all the indices $i \in S_j \setminus{S_j'}$ as $0 = f(\phi) \leq f(\{i\}) \leq f(S_j \setminus{S_j'})$ from the non-decreasing nature of $f$. Since, $\bigcup_{j \in [1:r]} (S_j \setminus S_j')=\{ i \in [1:n]  : |\{S \in \mathcal{F} : i\in S\}| > k \}$ and the fact that $\mathcal{F}$ satisfies the assumption that no two indices appear together, each such $i$ in $\{ i \in [1:n]  : |\{S \in \mathcal{F} : i\in S\}| > k \}$ would appear at least once in \eqref{eq: nSLES_7} for some $j \in [1:r]$. This lets us conclude that $f(\{i\})=0$ for all $i \in [1:n]  : |\{S \in \mathcal{F} : i\in S\}| > k$. Further, we get
\begin{align}
    f(Z) = 0, \forall Z \subseteq \{ i \in [1:n] : |\{S \in \mathcal{F} : i\in S\}| > k \}, \label{eq: nSLES_11}
\end{align}
as $0 = f(\phi) \leq f(Z)$ by the non-decreasing nature of $f$ and $f(Z) \leq \sum_{i \in Z} f(\{i\})$ by submodularity of $f$.

To show the modularity of $f$ on $2^{[1:n]}$, fix any $S \subseteq [1:n]$. Now, for any $\ell \in [1:m]$,
\begin{align}
    f(S \cap G_{\ell}) & = f((S \cap G_{\ell}) \cap R, (S \cap G_{\ell}) \cap R^{\text{c}}) \\
    & \leq f((S \cap G_{\ell}) \cap R) + f((S \cap G_{\ell}) \cap R^{\text{c}}) \label{eq: nSLES_9} \\
    & = f((S \cap G_{\ell}) \cap R), \label{eq: nSLES_10}
\end{align}
where \eqref{eq: nSLES_9} follows by submodularity of $f$ and \eqref{eq: nSLES_10} follows from \eqref{eq: nSLES_11}. Since $f$ is non-decreasing, $f((S \cap G_{\ell}) \cap R) \leq f(S \cap G_{\ell})$. Thus, $f(S \cap G_{\ell}) = f((S \cap G_{\ell}) \cap R)$. We also have that $|(S \cap G_{\ell}) \cap R| \leq 1$ because there can be at most one index in $G_{\ell}$ with exactly $k$ appearances in $\mathcal{F}$. So, $f((S \cap G_{\ell}) \cap R) = 0$ when $(S \cap G_{\ell}) \cap R = \phi$ and $f((S \cap G_{\ell}) \cap R) = f(t_{\ell})$ when $(S \cap G_{\ell}) \cap R = t_{\ell}$ for some $t_{\ell} \in [1:n]: t_{\ell} \in R$. Then,
\begin{align}
    f(S) & = f(S \cap G) + \sum\limits_{\ell = 1}^m f(S \cap G_{\ell}) \label{eq: nSLES_8} \\
    & = \sum\limits_{i \in S \cap G} f(\{i\}) + \sum\limits_{\substack{\ell \in [1:m]:\\ (S \cap G_{\ell}) \cap R = t_{\ell}}} f(t_{\ell}) \\
    & = \sum\limits_{\substack{i \in S:\\i \in R}} f(\{i\}) + \sum\limits_{\substack{i \in S:\\i \in R^{\text{c}}}} f(\{i\}) \label{eq: nSLES_12}\\
    & = \sum\limits_{i \in S} f(\{i\}),
\end{align}
where \eqref{eq: nSLES_8} follows from Lemma~\ref{lemma1} and \eqref{eq: nSLES_12} follows because $f(\{i\}) = 0$ for all $i \in S: i \in R^{\text{c}}$ by \eqref{eq: nSLES_11}. This concludes the proof of Claim~\ref{Shearer_Lemma_Equality_Submodular}.
\end{proof}

We are now ready to prove part 2) of Theorem~\ref{Equality_Implies_Modularity_Partition}. Suppose $\text{Gap}_{\text{U}}(f,\mathcal{F},\alpha)=0$. Let $\mathcal{F} = \{\!\!\{S_1, \ldots, S_r\}\!\!\}$. Let $\alpha(F_i)=\frac{p_i}{q_i}$ as $\alpha(F_i)$ is rational, for $i\in[1:r]$. Denoting $L=\textbf{lcm}(q_1,q_2,\dots,q_r)$, i.e., the least common multiple of $q_1,q_2,\dots,q_r$, we can express $\alpha(F_i)$ as $\frac{a_i}{L}$, for some $a_i\in\mathbb{N}$, $i\in[1:r]$.
Note that
\begin{align}
    \sum\limits_{S \in \mathcal{F}} \alpha(S)f(S) & = \sum\limits_{S \in \mathcal{F}} \frac{a_i}{L}f(S) \\
    & = \frac{1}{L} \sum\limits_{S \in \mathcal{F}} \sum\limits_{j = 1}^{a_i} f(S) \\
    & = \frac{1}{L} \sum\limits_{S \in \mathcal{F}'} f(S)\label{eq: SLES_5},
\end{align}
where \begin{align}
\mathcal{F}' = \{\!\!\{ \underbrace{S_1, \ldots, S_1}_{a_1 \text{ times}}, \underbrace{S_2, \ldots, S_2}_{a_2 \text{ times}}, \ldots, \underbrace{S_r, \ldots, S_r}_{a_r \text{ times}}\}\!\!\}.
\end{align}

Also, since
\begin{align}
    \sum_{S\in\mathcal{F}:i\in S}\alpha(S)&=\sum_{S\in\mathcal{F}:i\in S}\frac{a_i}{L}\\
    &=\sum_{S\in\mathcal{F}:i\in S}\sum_{j=1}^{a_i}\frac{1}{L}\\
    &=\sum_{S\in\mathcal{F}':i\in S}\frac{1}{L}\\
    &=\frac{1}{L}|\{S\in\mathcal{F}':i\in S\}|,
\end{align}
we have,
\begin{align}
    \{i:\sum_{S\in\mathcal{F}:i\in S}\alpha(S)=1\}&=\left\{i:|\left\{S\in\mathcal{F}':i\in S\right\}|=L\right\},\\
     \{i:\sum_{S\in\mathcal{F}:i\in S}\alpha(S)>1\}&=\left\{i:|\left\{S\in\mathcal{F}':i\in S\right\}|>L\right\}.
\end{align}

Now, by putting together \eqref{eq: SLES_5} with the fact that $\text{Gap}_{\text{U}}(f,\mathcal{F},\alpha)=0$, we get that
\begin{align}
    Lf([1:n])=\sum_{S\in\mathcal{F}'}f(S),
\end{align}
which along with Claim~\ref{Shearer_Lemma_Equality_Submodular} further implies that $f$ is modular and satisfies $f(Z) = 0, \forall Z \subseteq  \{i:\sum_{S\in\mathcal{F}:i\in S}\alpha(S)>1\}$.

For the other direction, suppose that $f$ is modular and satisfies $f(Z) = 0, \forall Z \subseteq  \{i:\sum_{S\in\mathcal{F}:i\in S}\alpha(S)>1\}=B$. Then, we have
\begin{align}
     &\textup{Gap}_{\textup{U}}(f, \mathcal{F}, \alpha) \\
     & = \sum_{S \in \mathcal{F}} \alpha(S) f(S)  - f([1:n]) \\
    & = \sum\limits_{S \in \mathcal{F}} \alpha(S) \sum\limits_{i \in S} f(\{i\}) - f([1:n]) \\
    & = \sum\limits_{i = 1}^n f(\{i\})\left( \sum\limits_{\substack{S \in \mathcal{F}: \\ i \in S}} \alpha(S) \right) - f([1:n]) \\
    & =\sum_{i\in B}f(\{i\})\left(\sum\limits_{\substack{S \in \mathcal{F}: \\ i \in S}} \alpha(S)-1\right)\nonumber\\
    &\hspace{12pt}+ \sum_{i\in B^\text{c}}f(\{i\})\left(\sum\limits_{\substack{S \in \mathcal{F}: \\ i \in S}} \alpha(S)-1\right)\\
    & = 0\label{eq: SLES_6},
\end{align}
where \eqref{eq: SLES_6} follows from the fact that $f(\{i\})=0$, for all $i\in B$ and $\sum\limits_{\substack{S \in \mathcal{F}:i \in S}} \alpha(S)=1$, for all $i\in B^\text{c}$.   This proves that $\text{Gap}_\text{U}(f, \mathcal{F}, \alpha) = 0$ if and only if $f$ is modular and $f(Z) = 0, \forall Z \subseteq  \{i:\sum_{S\in\mathcal{F}:i\in S}\alpha(S)>1\}$. The assertion for $\text{Gap}_\text{L}(f, \mathcal{F}, \beta)$ can be proved using this together with duality (Theorem~\ref{Duality}), as in the proof of Theorem~\ref{Stability_Submodular_Function_Partition}. Specifically, note that for every $i\in[1:n]$, we have 
\begin{align*}
\sum\limits_{S \in \mathcal{F}: i \in S} \beta(S)<1 \ \text{if and only if}\ \sum\limits_{S^{\text{c}}\in \bar{\mathcal{F}}:i\in S^{\text{c}}}\bar{\beta}(S^{\text{c}})>1.
\end{align*}
This completes the proof of Theorem~\ref{Equality_Implies_Modularity_Partition}.

\section{Proof of Lemma~\ref{lemma1}}\label{appendix:l}
Fix $S,S_1,\ldots,S_m$ arbitrarily, such that $S \subseteq G, S_1 \subseteq G_1, \ldots, S_m \subseteq G_m$. Then, for $H = G, \mathcal{T} = [1:m]$, we have,
\begin{align}
    0 & = \sum\limits_{i \in G} g(\{i\}) + \sum\limits_{\ell = 1}^m g(G_{\ell}) - g\bigg(\bigcup_{\ell = 1}^m G_{\ell}, G\bigg) \\
    & = \sum\limits_{i \in S} g(\{i\}) + \sum\limits_{i \in G \setminus{S}} g(\{i\}) + \sum\limits_{\ell = 1}^m g(S_{\ell}) \nonumber \\
    & \quad\quad +\sum\limits_{\ell = 1}^m g((G_{\ell} \setminus{S_{\ell}}) | S_{\ell}) - g\bigg(\bigcup_{\ell = 1}^m S_{\ell}, S \bigg) \nonumber \\
    & \quad\quad - g\bigg(G \setminus{S}, \Big( \bigcup_{\ell = 1}^m G_{\ell} \setminus{S_{\ell}} \Big) \Big| \bigcup_{\ell = 1}^m S_{\ell}, S \bigg) \label{eq: lemma1.1}\\
    & = \underbrace{\sum\limits_{i \in S} g(\{i\}) + \sum\limits_{\ell = 1}^m g(S_{\ell}) - g\bigg(\bigcup_{\ell = 1}^m S_{\ell}, S \bigg) }_{\text{I } \geq 0} \nonumber \\
    & \quad\quad + \underbrace{\sum\limits_{i \in G \setminus{S}} g(\{i\}) - g\bigg( (G\setminus{S}) \Big| \bigcup_{\ell = 1}^m S_{\ell}, S \bigg)}_{\text{II } \geq 0} \nonumber \\
    & \quad\quad + \underbrace{\sum\limits_{\ell = 1}^m g((G_{\ell} \setminus{S_{\ell}}) | S_{\ell}) - g\bigg(\bigcup_{\ell = 1}^m G_{\ell} \setminus{S_{\ell}} \Big| \bigcup_{\ell = 1}^m S_{\ell}, G \bigg)}_{\text{III } \geq 0}, \label{eq: lemma1.2}
\end{align}
where \eqref{eq: lemma1.1} follows from \eqref{eqref: chain_rule}. Notice that all three expressions in \eqref{eq: lemma1.2} are non-negative. In particular, I is non-negative by submodularity, II is non-negative as $\sum\limits_{i \in G \setminus{S}} g(\{i\}) \geq g\big(G \setminus{S}) \geq g\bigg((G \setminus{S}) \Big| \bigcup\limits_{\ell = 1}^m S_{\ell}, S\bigg)$ by submodularity and \eqref{eqref: conditioning}, and III is non-negative as $\sum\limits_{\ell=1}^m g((G_{\ell} \setminus{S_{\ell}}) | S_{\ell}) \geq \sum\limits_{\ell = 1}^m g\bigg((G_{\ell} \setminus{S_{\ell}}) \Big| \bigcup\limits_{\ell = 1}^m S_{\ell}, G\bigg) \geq g\bigg(\bigcup\limits_{\ell = 1}^m G_{\ell} \setminus{S_{\ell}} \Big| \bigcup\limits_{\ell = 1}^m S_{\ell}, G\bigg)$ by \eqref{eqref: conditioning} and submodularity. Since the sum of these three expressions is zero, each of these terms must be zero. This shows that expression I is equal to zero, and thus proves the required result.

\section{Example Motivating a Stronger Condition on $f$ in Part 2) of Theorem~\ref{Equality_Implies_Modularity_Partition}}\label{appendix:d}
    Let $f: 2^{[1:3]} \rightarrow \mathbb{R}$ be a set function with the following values: $f(\phi) = 0, ~f(\{1\}) = -100, ~f(\{2\}) = 0.001, ~f(\{3\}) = 50.0005, ~f(\{1,2\}) = -100, ~f(\{2,3\}) = 50.0005, ~f(\{1,3\}) = -50.0005, ~f(\{1,2,3\}) = -100$. Consider $\mathcal{F} = \{\!\!\{ \{1,2\}, \{1,2\}, \{2,3\}, \{1,3\} \}\!\!\}$ and $\alpha(S) = 1/2,~\forall S \in \mathcal{F}$. The function $f$ is a submodular but not modular and satisfies $\text{Gap}_{\text{U}}(f, \mathcal{F}, \alpha) = 0$. It also satisfies the condition that $f([1:j])$ is non-decreasing in $j$ as $f(\{1\}) = f(\{1,2\}) = f(\{1,2,3\}) = -100$. However, $f(S) > f(T)$ for some $S \subseteq T$ (specifically, for $S = \phi,\ \text{and}\ T = \{1,2\}$).

\section{Proof of Corollary \ref{Stability_Entropy_Partition}}\label{appendix:e}
Since entropy is a submodular function, we invoke Theorem~$\ref{Stability_Submodular_Function_Partition}$ with $f(S)=\mathsf{e}(S)$, $S\subseteq [1:n]$. 
Then $\text{Gap}_\text{L}(\mathsf{e}, \mathcal{F}, \gamma) \leq \varepsilon$ implies that, $\forall i \in [1:n]$,
\begin{align}
    H(X_i) + H(X_{[1:n] \setminus{\{i\}}}) - H(X_{[1:n]}) \leq \frac{\varepsilon}{\sigma}, 
\end{align}
which further implies that 
\begin{align}
     I(X_i ; X_{[1:n] \setminus{\{i\}}}) \leq \frac{\varepsilon}{\sigma}.
\end{align}
Similarly, we can show the above result for $\text{Gap}_\text{U}(\mathsf{e}, \mathcal{F}, \gamma)$.

\section{Proof of Corollary \ref{Equality_Implies_Independence}}\label{appendix:f}
\emph{Proof of Part 1).} we invoke part 1) of Theorem~\ref{Equality_Implies_Modularity_Partition} with $f(S)=\mathsf{e}(S)$, $S\subseteq [1:n]$. This implies that $\text{Gap}_\text{U}(\mathsf{e}, \mathcal{F}, \gamma) = 0$ if and only if $\mathsf{e}$ is modular. Notice that the modularity of $\mathsf{e}$ is equivalent to the condition 
 $H(X_{[1:n]}) = \sum\limits_{i = 1}^n H(X_i)$, as this equality holds if and only if $X_i$, $i\in[1:n]$ are mutually independent~\cite[Theorem~2.6.6]{thomas2006elements}. 

\emph{Proof of Part 2).} Since entropy is a non-decreasing submodular function, we invoke part 2) of Theorem~\ref{Equality_Implies_Modularity_Partition} with it. This implies that $\text{Gap}_\text{U}(\mathsf{e}, \mathcal{F}, \alpha) = 0$ if and only if $\mathsf{e}$ is modular and $\mathsf{e}(Z)=0$, for all $Z\subseteq\{i\in[1:n]: \sum_{F \in \mathcal{F}: i \in F} \alpha(F) > 1 \}$.
From the arguments in the proof of part 1) of this corollary, we know that $\mathsf{e}$ is modular if and only if $X_i$, $i \in [1:n]$ are mutually independent.
Moreover, consider an arbitrary $i$ such that $\sum\limits_{F\in\mathcal{F}:i\in F}\alpha(F)>1$. For this $i$, $\mathsf{e}(\{i\}) = H(X_i) = 0$ which occurs if and only if $X_i$ is a constant.

\section{Proof of Corollary \ref{Relative_Entropy_Partition_Equality}}\label{appendix:g}
\emph{Proof of Part 1).} From Theorem \ref{Equality_Implies_Modularity_Partition}, we have that
\begin{align}
    \sum\limits_{F \in \mathcal{F}} \gamma(F) D(P_{X_F} || Q_{X_F}) ~ = ~ D(P_{X_{[1:n]}} || Q_{X_{[1:n]}})
\end{align}
if and only if $d$ is modular, i.e., $\forall F \subseteq [1:n], ~D(P_{X_{F}} || Q_{X_{F}}) = \sum\limits_{i \in F} D(P_{X_i} || Q_{X_i})$. Consider $F = [1:n]$, we have,
\begin{align}
 D(P_{X_{[1:n]}} || Q_{X_{[1:n]}}) = \sum\limits_{i = 1}^n D(P_{X_i} || Q_{X_i}) 
\end{align}
which implies that
\begin{align}
     & H_P({X_{[1:n]}}) - \sum\limits_{x_{[1:n]}} P_{X_{[1:n]}}(x_{[1:n]}) \log{Q_{X_{[1:n]}}}(x_{[1:n]}) \nonumber \\
    &\hspace{12pt} = - \sum\limits_{i = 1}^{n} H_P({X_i}) - \sum\limits_{i = 1}^n\sum\limits_{x_i} P_{X_i}(x_i)\log{Q_{X_i}(x_i)},
\end{align}
(where $H_P(.)$ denotes entropy with respect to the distribution $P_{X_{[1:n]}}$) which further implies that 
\begin{align}
    H_P({X_{[1:n]}}) = \sum\limits_{i = 1}^{n} H_P({X_i}) \label{eq: REPE_1}
\end{align}
because $P_{X_{[1:n]}}(x_{[1:n]}) \log{Q_{X_{[1:n]}}}(x_{[1:n]}) = \sum\limits_{i = 1}^n\sum\limits_{x_i} P_{X_i}(x_i)\log{Q_{X_i}(x_i)}$ as $Q_{X_{[1:n]}}$ is a product probability distribution. Now, \eqref{eq: REPE_1} holds if and only if $P_{X_{[1:n]}}$ is a product probability distribution~\cite[Theorem~2.6.6]{thomas2006elements}.

\emph{Proof of Part 2).} Since relative entropy is a non-decreasing submodular function, we invoke part 2) of Theorem~\ref{Equality_Implies_Modularity_Partition} with it. This implies that $\text{Gap}_\text{U}(d, \mathcal{F}, \alpha) = 0$ if and only if $d$ is modular and $d(Z)=0$, for all $Z\subseteq\{i\in[1:n]: \sum_{F \in \mathcal{F}: i \in F} \alpha(F) > 1 \}$.
From the arguments in the proof of part 1) of this corollary, we know that $d$ is modular if and only if $P_{X_{[1:n]}}$ is a product probability distribution.
Moreover, consider $Z=\{i\in[1:n]:\sum\limits_{F\in\mathcal{F}:i\in F}\alpha(F)>1\}$. We have that $d(Z) = 0$ which occurs if and only if the distributions $P_{X_{Z}}$ and $Q_{X_{Z}}$ are equal. 

\section{Proof of Corollary~\ref{Matroid_Rank_Equality}}\label{appendix:h}

From part 1) of Theorem \ref{Equality_Implies_Modularity_Partition}, we have that 
\begin{align}
    \sum\limits_{F \in \mathcal{F}} \gamma(F) r(F) = r(\mathcal{E})
\end{align}
if and only if $r$ is modular, i.e., $\forall F \subseteq \mathcal{E}, ~r(F) = \sum_{i \in F} r(\{i\})$. Let $\mathcal{M} = (\mathcal{E}, \mathcal{I})$ be any matroid with the modular rank function $r$. We show that $\mathcal{I} = 2^{B^{\text{c}}}$, where $B = \{ j : \{j\} \cap I = \phi, ~\forall I \in \mathcal{I} \}$. We first prove that if $F \subseteq B^{\text{c}}$, then $F \in \mathcal{I}$. Note that for any arbitrary $F \subseteq \mathcal{E}$, we have
\begin{align}
    r(F \cap B^{\text{c}}) & = \sum\limits_{\substack{i \in F, \\ i \notin B}} r(i) \\
    & = \sum\limits_{\substack{i \in F, \\ i \notin B}} 1\\
    &= |F \cap B^{\text{c}}|, \label{eq: MRE_1} 
\end{align}
where \eqref{eq: MRE_1} follows because $\forall i \in B^{\text{c}}$, $r(\{i\})=1$ from the definitions of matroid and rank function. Now \eqref{eq: MRE_1} implies that $F \cap B^{\text{c}} \in \mathcal{I}$, $\forall F \subseteq \mathcal{E}$ from the property of matroid that $r(T) = |T| \implies T \in \mathcal{I} \text{ for any } T \subseteq \mathcal{E}$. This further implies that $F \in \mathcal{I}, ~\forall F \subseteq B^{\text{c}}$.

We conclude by showing that if $F \in \mathcal{I}$, then $F \subseteq B^{\text{c}}$. Assume $F \in \mathcal{I}$. From the definition of matroid rank function, we have that $r(F) = |F|$. Also, we have
\begin{align}
    |F|&=r(F)\\
    &= \sum\limits_{\substack{i \in F, \\ i \notin B}} r(i) + \sum\limits_{\substack{i \in F, \\ i \in B}} r(i) \\
    & = |F \cap B^{\text{c}}|, \label{eq: MRE_4} 
\end{align}
where \eqref{eq: MRE_4} follows from the fact that $r(i) = 0,~\forall i \in B$. Now \eqref{eq: MRE_4} implies that $F\subseteq B^{\text{c}}$. This proves that $\mathcal{I} = 2^{B^{\text{c}}}$.

\section{Proof of Proposition~\ref{prop:MMI-specialcases}}\label{appendix:i}
\emph{Proof of Part 1).} The non-negativity follows directly from Theorem~\ref{MT inequality}, noting that entropy is a submodular function. From Part 1) of Corollary~\ref{Equality_Implies_Independence}, we have $(\mathcal{F}, \gamma)\text{-}MI(X_1, \cdots, X_n) = 0$ if and only if $X_1, \ldots, X_n$ are mutually independent.

\emph{Proof of Part 2).} 
For $\mathcal{F} = \{\!\!\{\{i\}: i \in [1:n]\}\!\!\}$ and $\gamma(F) = 1, ~\forall F \in \mathcal{F}$, we have
\begin{align*}
    \sum\limits_{F \in \mathcal{F}} \gamma(F)H(X_F) - H(X_{[1:n]}) & = \sum\limits_{i = 1}^n H(X_i) - H(X_{[1:n]}) \\
    & = \text{TC}(X_1; \cdots; X_n).
\end{align*}

\emph{Proof of Part 3).} 
For $\mathcal{F} = \{\!\!\{ [1:n] \setminus{\{i\}} : i \in [1:n] \}\!\!\}$ and $\gamma(F) = 1/(n-1), ~\forall F \in \mathcal{F}$, we have
\begin{align}
    & \sum\limits_{F \in \mathcal{F}}
    \gamma(F)H(X_F) - H(X_{[1:n]}) \\
    & = \sum\limits_{i = 1}^n \frac{1}{n-1} H(X_{[1:n] \setminus{\{i\}}}) - H(X_{[1:n]}) \\
    & = \frac{1}{n-1} \sum\limits_{i = 1}^n \left( H(X_{[1:n]}) - H(X_{\{i\}} | X_{[1:n] \setminus{\{i\}}}) \right) - H(X_{[1:n]}) \label{eq: prop1_3.1} \\
    & = \frac{1}{n-1} \left( H(X_{[1:n]}) - \sum\limits_{i = 1}^n H(X_{\{i\}} | X_{[1:n] \setminus{\{i\}}}) \right) \\
    & = \frac{1}{n-1} \text{DTC}(X_1; \cdots; X_n),
\end{align}
where \eqref{eq: prop1_3.1} follows from chain rule for entropy.

\emph{Proof of Part 4).} \begin{align}
    & \min_{\gamma:\mathcal{F}\rightarrow \mathbb{R}_+}\frac{(\mathcal{F},\gamma)\text{-}MI(X_1;\cdots;X_n)}{w(\gamma)-1} \\
    & = \min_{\gamma:\mathcal{F}\rightarrow \mathbb{R}_+} \text{Gap}_{\text{L}}(\mathsf{e}, \mathcal{F}, \bar{\gamma}) \label{eq: prop1_4.1} \\
    & = \min_{\gamma:\mathcal{F}\rightarrow \mathbb{R}_+} H(X_{[1:n]}) - \sum_{F \in \mathcal{F}} \bar{\gamma}(F) H(X_F | X_{F^\text{c}}) \\
    & = \min_{\gamma:\mathcal{F}\rightarrow \mathbb{R}_+} H(X_{[1:n]}) - \sum_{F \in \mathcal{F}} \gamma(F) H(X_F | X_{F^\text{c}}) \label{eq: prop1_4.2}\\
    & = H(X_{[1:n]}) - \max_{\gamma:\mathcal{F}\rightarrow \mathbb{R}_+} \sum_{F \in \mathcal{F}} \gamma(F) H(X_F | X_{F^\text{c}}) \\
    & = \text{SI}(X_1;\cdots; X_n),
\end{align}
where \eqref{eq: prop1_4.1} follows from duality and the fact that $\bar{\mathcal{F}} = \mathcal{F}$ for $\mathcal{F}=2^{[1:n]}\setminus\{\phi,[1:n]\}$, and \eqref{eq: prop1_4.2} follows from one-to-one correspondence between the fractional partitions $\gamma$ and $\bar{\gamma}$, i.e.,
\begin{align*}
    \bar{\gamma}(F^{\text{c}}) & = \frac{\gamma(F)}{w(\gamma) - 1}, \\ 
    \text{and \quad} \gamma(F) & = \frac{\bar{\gamma}(F^{\text{c}})}{w(\bar{\gamma}) - 1}.
\end{align*}

\section{Proof of Theorem~\ref{Multivariate_Mutual_Information}}\label{appendix:j}

We first present the proof of Property 1). Note that
\begin{align}
    (\mathcal{F}, \gamma) & \text{-}MI(X_1, \cdots, X_n) \nonumber \\
    & = \sum\limits_{F \in \mathcal{F}} \gamma(F)H(X_F) - H(X_{[1:n]}) \\
    & \leq \sum\limits_{F \in \mathcal{F}} \gamma(F)\sum\limits_{i \in F} H(X_i) - \sum\limits_{i = 1}^n H(X_i | X_{[1:i-1]}) \label{eq: MMI_2.1}\\
    & = \sum\limits_{i = 1}^n H(X_i) \left( \sum\limits_{\substack{F \in \mathcal{F}: \\ i \in F}} \gamma(F) \right) - \sum\limits_{i = 1}^n H(X_i | X_{[1:i-1]}) \label{eq: MMI_2.2} \\
    & = \sum\limits_{i = 1}^n H(X_i) -  \sum\limits_{i = 1}^nH(X_i | X_{[1:i-1]}) \label{eq: MMI_2.3} \\
    & = \sum\limits_{i = 1}^n H(X_i) - H(X_{[1:n]})\\
    &= \text{TC}(X_1, \cdots, X_n),
\end{align}
where \eqref{eq: MMI_2.1} follows from the subadditivity of entropy, \eqref{eq: MMI_2.2} follows from the interchange of summations, \eqref{eq: MMI_2.3} follows from the fact that $\sum\limits_{\substack{F \in \mathcal{F}: \\ i \in F}} \gamma(F) = 1$ and the chain rule for entropy. Finally, this bound is achieved when $\mathcal{F} = \{\!\!\{ \{i\}: i \in [1:n]\}\!\!\}$ and $\gamma(F) = 1, ~\forall F \in \mathcal{F}$.

We now present the proof of Property 3) and we use it to prove Property 2) later. Let $\mathcal{F} = \{\!\!\{ F_1, \ldots, F_p\}\!\!\}$ and correspondingly let $\tilde{\mathcal{F}} = \{\!\!\{ \tilde{F}_1, \ldots, \tilde{F}_p\}\!\!\}$. From the statement of Property 3) we also get that $\gamma(F_i) = \tilde{\gamma}(\tilde{F}_i)$ for all $i\in[1:p]$. Then,
\begin{align}
    (\mathcal{F}, & \gamma) \text{-}MI(X_1, \cdots, X_n) - (\mathcal{\tilde{F}}, \tilde{\gamma})\text{-}MI(X_1, \cdots, X_{n-1}) \\
    & = \sum\limits_{F \in \mathcal{F}} \gamma(F)H(X_F) - H(X_{[1:n]}) \nonumber \\
    & \quad \quad \quad - \sum\limits_{F \in \mathcal{\tilde{F}}} \tilde{\gamma}(\tilde{F})H(X_{\tilde{F}}) + H(X_{[1:n-1]}) \\
    & = \sum\limits_{i = 1}^p \left[\gamma(F_i)H(X_{F_i}) - \tilde{\gamma}(\tilde{F_i})H(X_{\tilde{F_i}}) \right] \nonumber \\
    & \quad \quad \quad - H(X_n | X_{[1:n-1]}) \label{eq: MMI_3.6}\\
    & = \sum\limits_{i = 1}^p \gamma(F_i) \left[ H(X_{F_i}) - H(X_{\tilde{F_i}}) \right] - H(X_n | X_{[1:n-1]})\\
    & = \sum\limits_{i = 1}^p \gamma(F_i) H(X_{F_i \setminus{\tilde{F_i}}} | X_{\tilde{F_i}}) - H(X_n | X_{[1:n-1]}) \label{eq: MMI_3.1} \\
    & = \sum\limits_{\substack{F \in \mathcal{F}: \\ n \in F}} \gamma(F) H(X_n | X_{F \setminus{\{n\}}}) - H(X_n | X_{[1:n-1]}) \label{eq: MMI_3.2} \\
    & = \sum\limits_{\substack{F \in \mathcal{F}: \\ n \in F}} \gamma(F) H(X_n | X_{F \setminus{\{n\}}}) \nonumber \\
    & \quad \quad - \sum\limits_{\substack{F \in \mathcal{F}: \\ n \in F}} \gamma(F) H(X_n | X_{[1:n-1]}) \label{eq: MMI_3.3} \\
    & = \sum\limits_{\substack{F \in \mathcal{F}: \\ n \in F}} \gamma(F) \left[H(X_n | X_{F \setminus{\{n\}}}) - H(X_n | X_{[1:n-1]}) \right] \label{eq: MMI_3.4} \\
    & = \sum\limits_{\substack{F \in \mathcal{F}: \\ n \in F}} \gamma(F) I(X_n; X_{[1:n-1] \cap F^{\text{c}}} | X_{[1:n-1] \cap F}), \label{eq: MMI_3.5}
\end{align}
where \eqref{eq: MMI_3.6} and \eqref{eq: MMI_3.1} follow from chain rule for entropy, and \eqref{eq: MMI_3.2} follows from the fact that $F_i \setminus{\tilde{F}_i}$ can only have two possible values, i.e., $\phi$ (when $F_i = \tilde{F}_i$) or $\{n\}$ (when $F_i \neq \tilde{F}_i$) as $\tilde{\mathcal{F}}$ is constructed from $\mathcal{F}$ by truncating $n$ in all $F \in \mathcal{F}$ where $n \in F$. We use the fact that $\sum_{\substack{F \in \mathcal{F}: n \in F}} \gamma(F) = 1$ to write \eqref{eq: MMI_3.3}. \eqref{eq: MMI_3.5} follows from the fact that $F \setminus{\{n\}} = [1:n-1] \cap F$ when $n \in F$. Finally, we can observe that if we lower bound $H(X_n | X_{F \setminus{\{n\}}})$ in the first term of each summand in \eqref{eq: MMI_3.4} by $H(X_n | X_{[1:n-1]})$ then we would get that $(\mathcal{F}, \gamma)\text{-}MI(X_1, \cdots, X_n) \geq (\mathcal{\tilde{F}}, \tilde{\gamma})\text{-}MI(X_1, \cdots, X_{n-1})$. On the contrary, if we were to upper bound $H(X_n | X_{F \setminus{\{n\}}})$ in the first term of each summand in \eqref{eq: MMI_3.4} by $H(X_n)$, then we would get that $(\mathcal{F}, \gamma)\text{-}MI(X_1, \cdots, X_n) \leq (\mathcal{\tilde{F}}, \tilde{\gamma})\text{-}MI(X_1, \cdots, X_{n-1}) + I(X_n; X_{[1:n-1]})$.

With this information, we will now prove Property 2). We first prove the inequality for the case when $Y_i = X_i, ~\forall i \in [1:n-1]$. Proof for the general inequality then follows from induction and symmetry arguments.
\begin{align}
    (\mathcal{F}, \gamma) & \text{-}MI(X_1, \cdots, X_{n-1}, Y_n) \nonumber \\
    & \quad \quad - (\mathcal{F}, \gamma)\text{-}MI(X_1, \cdots, X_{n-1}, X_n) \label{eq: MMI_4.1}\\
    & = \sum\limits_{\substack{F \in \mathcal{F}: \\ n \in F}} \gamma(F) I(Y_n; X_{[1:n-1] \cap F^{\text{c}}} | X_{[1:n-1] \cap F}) \nonumber \\
    & \quad \quad - \sum\limits_{\substack{F \in \mathcal{F}: \\ n \in F}} \gamma(F) I(X_n; X_{[1:n-1] \cap F^{\text{c}}} | X_{[1:n-1] \cap F}) \label{eq: MMI_4.2} \\
    & \leq \sum\limits_{\substack{F \in \mathcal{F}: \\ n \in F}} \gamma(F) I(Y_n, X_n; X_{[1:n-1] \cap F^{\text{c}}} | X_{[1:n-1] \cap F}) \nonumber \\ 
    & \quad \quad - \sum\limits_{\substack{F \in \mathcal{F}: \\ n \in F}} \gamma(F) I(X_n; X_{[1:n-1] \cap F^{\text{c}}} | X_{[1:n-1] \cap F}) \label{eq: MMI_4.3} \\
    & = \sum\limits_{\substack{F \in \mathcal{F}: \\ n \in F}} \gamma(F) I(Y_n; X_{[1:n-1] \cap F^{\text{c}}} | X_{[1:n-1] \cap F}, X_n) \label{eq: MMI_4.4} \\
    & \leq \sum\limits_{\substack{F \in \mathcal{F}: \\ n \in F}} \gamma(F) H(Y_n | X_{[1:n-1] \cap F}, X_n) \label{eq: MMI_4.5} \\
    & \leq \sum\limits_{\substack{F \in \mathcal{F}: \\ n \in F}} \gamma(F) H(Y_n | X_n) \label{eq: MMI_4.6} \\
    & = H(Y_n | X_n), \label{eq: MMI_4.7}
\end{align}
where \eqref{eq: MMI_4.2} follows from \eqref{eq: MMI_3.5}, and \eqref{eq: MMI_4.3} and \eqref{eq: MMI_4.4} follow from the application of chain rule for mutual information, \eqref{eq: MMI_4.6} follows from the fact that conditioning only reduces entropy, and \eqref{eq: MMI_4.7} uses the fact that $\sum\limits_{\substack{F \in \mathcal{F}: \\ n \in F}} \gamma(F) = 1$.

Proof for Property 4) follows directly by invoking Han's result on the symmetry property of entropy vectors in the entropy vector space \cite[Lemma~3.1]{Han78} and by noting that $\gamma_n = -1$. 

\section{Proof of Proposition~\ref{Determinantal_Equality}}\label{appendix:k}

Let $X_{[1:n]}$ have a multivariate Gaussian distribution with mean $0$ and covariance matrix $K$, then
\begin{align}
    h(X_{[1:n]}) & = \frac{1}{2}\log{[(2\pi e)^n |K|]}, \\
    \text{and \quad} h(X_F) & = \frac{1}{2}\log{[(2\pi e)^{|F|} |K(F)|]}.
\end{align}
Now, \cite[Proof of Corollary~III]{MadimanT10} in conjunction with part 1) of Corollary~\ref{Equality_Implies_Independence} (applied to differential entropy) implies that $ \prod\limits_{F \in \mathcal{F}} |K(F)|^{\gamma(F)} = |K| $ if and only if $X_i$, $i\in[1:n]$ are mutually independent, which is equivalent to $K$ being a diagonal matrix, i.e., $K_{ij}=0$, for all $i\neq j$.

\fi



\end{document}